\documentclass[review]{elsarticle}

\usepackage{hyperref}
\usepackage{amssymb,amsmath,amsthm}
\usepackage{epsfig}
\usepackage{epstopdf}
\usepackage{subfigure}
\usepackage{multirow}
\usepackage{lscape} 
\usepackage{multirow}
\usepackage{booktabs}

\usepackage{color}

\newtheorem{proposition}{Proposition}[section]

\newtheorem{definition}{Definition}[section]
\newtheorem{example}{Example}[section]

\journal{Journal of \LaTeX\ Templates}

\biboptions{authoryear}

\begin{document}

\begin{frontmatter}

\title{Discriminant Analysis of Distributional Data via Fractional Programming}


\author[dias]{S\'onia Dias \corref{mycorrespondingauthor}}
\cortext[mycorrespondingauthor]{Corresponding author}
\ead{sdias@estg.ipvc.pt}
\author[bri]{Paula Brito}
\ead{mpbrito@fep.up.pt}
\author[ama]{Paula Amaral}
\ead{paca@fct.unl.pt}
\address[dias]{Instituto Polit\'{e}cnico de Viana do Castelo \& LIAAD-INESC TEC, Portugal}
\address[bri]{Faculdade de Economia,  Universidade do Porto \& LIAAD-INESC TEC, Portugal}
\address[ama]{CMA \& Faculty of Science and Engineering, Universidade Nova de Lisboa, Portugal}

\begin{abstract}
We address classification of distributional data, where units are described by histogram or interval-valued variables. The proposed approach uses a linear discriminant function  where distributions or intervals are represented by quantile functions, under specific assumptions. This discriminant function allows defining a score for each unit, in the form of a quantile function, which is used to classify the units in two a priori groups, using the Mallows distance. There is a diversity of application areas for the proposed linear discriminant method. In this work we classify the airline companies operating in NY airports based on air time and arrival/departure delays, using a full year flights.
\end{abstract}


\end{frontmatter}

\section{Introduction}
\label{s1}

The type of data that is necessary to analyse has changed over the past few years. The evolution of technology allows recording large sets of data, it is therefore necessary to develop efficient and more precise methods to analyse such big data. The large size of  datasets often leads to the need to aggregate units and build macrodata. As a consequence, data analysts are confronted with data where the observations are not single values or categories but finite sets of values/categories, intervals or distributions. It is therefore necessary to develop new methods that allow analysing these more complex datasets. Symbolic Data Analysis \citep{bodi00,bidi07,noirbri11,br14} is a recent statistical framework that studies and develops methods for such kind of symbolic variables. As in classical statistics, these new variables can be classified as quantitative or qualitative.  Depending on the type of realization, quantitative variables may be single-valued  - when each unit is allowed to take just one single value, multi-valued - when each unit is allowed to take a finite set of values, interval-valued - when an interval of real values is recorded for each unit or modal-valued - when each unit is described by a probability/frequency/weight distribution. Histogram-valued variables, that we are considering in this work, are a particular type of modal-valued variables \citep{bodi00,bidi07,noirbri11,br14}; we note that interval-valued variables are a special case of these, so that the developed models and methods also apply to interval data.

\begin{example} \label{ex1}
Consider data about the Air Time and Arrival Delay of all flights of five airlines departing from a given airport. Here, the entities of interest are not the individual flights but the airlines, for each of which we aggregate information. The values of variables Air Time and Arrival Delay may be aggregated in the form of interval or histogram-valued variables; in the first case each airline is represented by an interval, defined by the minimum and maximum observed values; for histogram-valued variables,  each airline is represented by the empirical distributions of the records associated with each variable.
Table \ref{table_AirTimeArrive } presents the obtained symbolic data, where Arrival Delay is an interval-valued variable and Air Time is a histogram-valued variable.
\end{example}
Linear models are frequently used in multivariate data analysis, as in linear regression and linear discriminant analysis. However, when the observations are not single values but intervals or distributions, the classical definition of linear combination has to be adapted. According to the definition proposed by \citet{dibr15}, the observations of the variables that are distributions or intervals of values are rewritten as quantile functions, the inverse of cumulative distribution functions \citep{irve15}, thereby taking into account the distribution within the intervals or subintervals that compose the histograms.
\begin{center}
	\begin{table}
\caption{Arrival Delay and Air Time interval and histogram-valued variables, respectively.}
{\scriptsize
\begin{tabular}{|c|c|c|}
\hline
Airline  & \multirow{2}{*}{Arrival Delay}  & \multirow{2}{*}{Air Time } \\
(IATA code) && \\
\hline
9E & $[-68,744]$ & $\{[21,56.5[,0.3;[56.5,106[,0.4;[106,196[,0.2;[196,272[,0.1\}$\\
EV & $[-62,577]$ & $\{[20,49[,0.2;[49,76[,0.2;[76,97[,0.2;[97,124[,0.2;[124,286],0.2\}$\\
MQ & $[-53,1127]$ & $\{[33,68[,0.2;[68,77[,0.2;[77,105[,0.3;[105,236],0.3\}$\\
OO & $[-26,157]$ & $\{[50,68[,0.4;[68,70[,0.2;[70,177[,0.4\}$\\
YV &  $[-46,381]$ & $\{[32,47[,0.2;[47,51[,0.2;[51,77[,0.2;[77,85[,0.2;[85,122],0.2\}$\\
\hline
\end{tabular}}
\label{table_AirTimeArrive }
\end{table}
\end{center}
 In the case of histogram-valued variables, the Uniform distribution is usually assumed; for interval-valued variables different distributions may be assumed \citep{dibr17}, up to this date  Uniform and Triangular distributions have been considered in linear regression models for this type of variables \citep{dibr17, tesemalaquias17}. The criterion to be optimized to define linear models, for both variable types, is based on the Mallows distance. Using the linear combination proposed in \citet{dibr15} we define, in this work, a linear discriminant function which allows defining for each unit a score that is a quantile function. The units are then classified based on the distance between its score and that of the barycenter of each a priori classes.

For interval-valued variables, parametric classification rules, based on Normal or Skew-Normal distributions, were proposed in \citet{sibri15}. 
Non-parametric methodologies for discriminant analysis of interval data may be found in, e.g. \cite{Ishibuchi90}, \cite{Nivlet01},  \cite{Rossi02},  \cite{DSBr06}, \cite{Carrizosa07}, \cite{Angulo07}, \cite{Lauro08},  \cite{Utkin11}.

The remaining of the paper is organized as follows. Section 2 introduces  histogram and interval-valued variables and their representations, the distance used to evaluate the similarity between histogram/intervals and the definition of linear combination for these types of variables. Section 3 presents the discriminant function and the optimization problem that allows obtaining the model parameters. Section 4 reports a simulation study and discusses its results. In Section 5, an application to flights data is presented. Finally, Section 6 concludes the paper, pointing out directions for future research.

\section{Concepts about histogram-valued variables} \label{s2}

In this section we introduce definitions and recall results needed to support the symbolic linear discriminant analysis that will be proposed and that allow for the classification of a set of units in two classes.

\subsection{Histogram-valued variables and their representations}\label{s2.1}

Histogram-valued variables \citep{bidi07, noirbri11}, are formally defined as follows.

\begin{definition}\label{def2.1}
     $Y$ is a histogram-valued variable when to each unit $i \in \{1,\ldots,n\}$ corresponds a histogram $Y(i)$ defined by a finite number of contiguous and non-overlapping intervals, each of which is associated with a (non-negative) weight. Then, $Y(i)$ can be represented by a histogram (\cite{bidi03}):
\begin{equation}\label{eqHinterval}
H_{Y(i)}=\left\{\left[\underline{I}_{Y(i)_1},\overline{I}_{Y(i)_1}\right[,p_{i1}; \left[\underline{I}_{Y(i)_2},\overline{I}_{Y(i)_2}\right[,p_{i2};\ldots;
\left[\underline{I}_{Y(i){m_{i}}},\overline{I}_{Y(i){m_{i}}}\right],p_{im_{i}}\right\}
\end{equation}
 $p_{i\ell}$ is the probability or frequency associated with the subinterval $\left[\underline{I}_{Y(i)_{\ell}},\overline{I}_{Y(i)_{\ell}}\right[,$ $\ell \in \left\{1,2,\ldots,m_{i}\right\},$ $m_{i}$  is the number of subintervals for   unit $i$; $\displaystyle\sum\limits_{\ell=1}^{m_{i}} p_{i\ell}=1;$ $\underline{I}_{Y(i)_\ell} \leq \overline{I}_{Y(i)_{\ell}}$ for $\ell \in \left\{1,2,\ldots,m_{i}\right\},$ and $\overline{I}_{Y(i)_{\ell-1}}\leq \underline{I}_{Y(i)_{\ell}},$  $\ell \in \left\{2,\ldots,m_{i}\right\}.$ 

\end{definition}

Each subinterval $I_{Y(i)_\ell}$ may be represented by its lower and upper bounds $\underline{I}_{Y(i)_\ell}$ and $\overline{I}_{Y(i)_\ell}$, or by its centre $c_{Y(i)_\ell}=\frac{\overline{I}_{Y(i)_\ell}+\underline{I}_{Y(i)_\ell}}{2}$ and half-range $r_{Y(i)_\ell}=\frac{\overline{I}_{Y(i)_\ell}-\underline{I}_{Y(i)_\ell}}{2}$.
$Y(i)$ may, alternatively, be represented by the inverse cumulative distribution function, the quantile function $\Psi_{Y(i)}^{-1}$, under specific assumptions \citep{irve06, dibr15}. 

Henceforth, in all representations, it is assumed that within each subinterval $\left[\underline{I}_{Y(i)_{\ell}},\overline{I}_{Y(i)_{\ell}}\right[$  the values for the variable $Y,$ for unit $i$, are uniformly distributed.
In this case, the quantile function is piecewise linear and is given by
\begin{equation}\label{eqHFQ}
\Psi_{Y(i)}^{-1}(t)=\left\{{\renewcommand{\arraystretch}{1.25}\begin{array}{lll}
                   c_{Y(i)_1}+\left(\frac{2t}{w_{i1}}-1\right) r_{Y(i)_1} & \textit{ if } & 0 \leq t < w_{i1} \\
                   c_{Y(i)_2} +\left(\frac{2(t-w_{i1})}{w_{i2}-w_{i1}}-1\right) r_{Y(i)_2} & \textit{ if } & w_{i1} \leq t < w_{i2} \\
                     \vdots & & \\
                    c_{Y(i)_{m_i}} +\left(\frac{2(t-w_{i(m_i-1)})}{1-w_{i(m_i-1)}}-1\right) r_{Y(i)_{im_i}} & \textit{ if } & w_{i(m_i-1)} \leq t \leq
                   1
 \end{array}}
  \right.
\end{equation}

\noindent where $w_{i\ell}=\displaystyle \sum_{h=1}^{\ell} p_{ih}$,  $\ell \in\{1,\ldots,m_{i}\}$, and $m_{i}$ is the number of subintervals in $Y(i).$

In the previous definition:
\begin{itemize}
\item The subintervals of histograms $H_{Y(i)}$ should be ordered and disjoint, if not they must be rewritten in that form (see \cite{tesedias14}, Appendix A).
\item  For different units, the number of subintervals, or pieces in the quantile functions, may be different. If necessary,  this function may be rewritten with the same number of pieces and the same domain for each piece for all units (see \cite{tesedias14}, Section 2.2.3).
\end{itemize}

When $m_i=1$ for each unit $i,$ $Y(i)$ is the interval $\left[\underline{I}_{Y(i)},\overline{I}_{Y(i)}\right[$  with $p_{i}=1;$ the histogram-valued variable is then reduced to an interval-valued variable. The corresponding quantile function is:
$\Psi_{Y(i)}^{-1}(t)=
 c_{Y(i)}+ \left(2t-1\right)r_{Y(i)}$,  $0 \leq t \leq 1$, where  $c_{Y(i)}= \frac{(\underline{I}_{Y(i)}+\overline{I}_{Y(i)})}{2}$ and $r_{Y(i)}= \frac{(\overline{I}_{Y(i)} -\underline{I}_{Y(i)})}{2}$.

\begin{example} \label{ex2}

Consider the distribution of variable Air Time for the airline 9E in Example \ref{ex1}:  $H_{9E}=\left\{[21,56.5[,0.3;[56.5,106[,0.4;[106,196[,0.2;[196,272[,0.1\right\}.$
This histogram may also be represented by the quantile function (see \textit{Figure \ref{figHist1}}):

\small{
$$
\Psi_{9E}^{-1}(t)=\left\{\begin{array}{lll}
                     38.75+\frac{t}{0.3}\times 18.75  & \quad if \quad  & 0 \leq t <0.3 \\
                    81.25 +\frac{t-0.3}{0.4} \times 25.75 &  \quad if \quad & 0.3 \leq t < 0.7 \\
                    151 +\frac{t-0.7}{0.2} \times 45 &  \quad if \quad & 0.7 \leq t < 0.9 \\
                   234 +\frac{t-0.9}{0.1} \times 38 & \quad if \quad & 0.9 \leq t \leq 1
 \end{array}
  \right.
  $$}

\begin{figure}[h!] 
  \centering
   \includegraphics[width=0.5\textwidth]{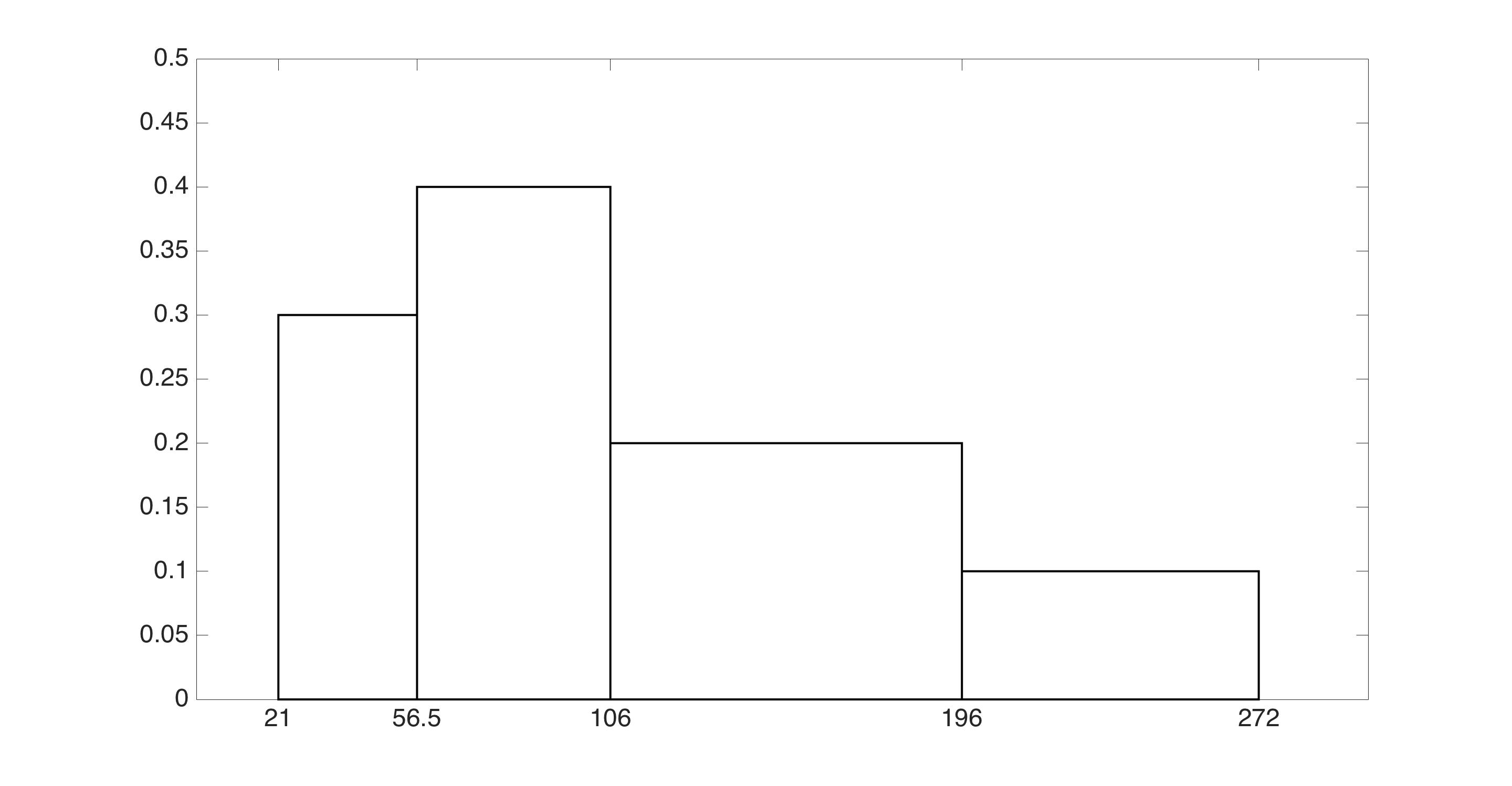}~
    \includegraphics[width=0.5\textwidth]{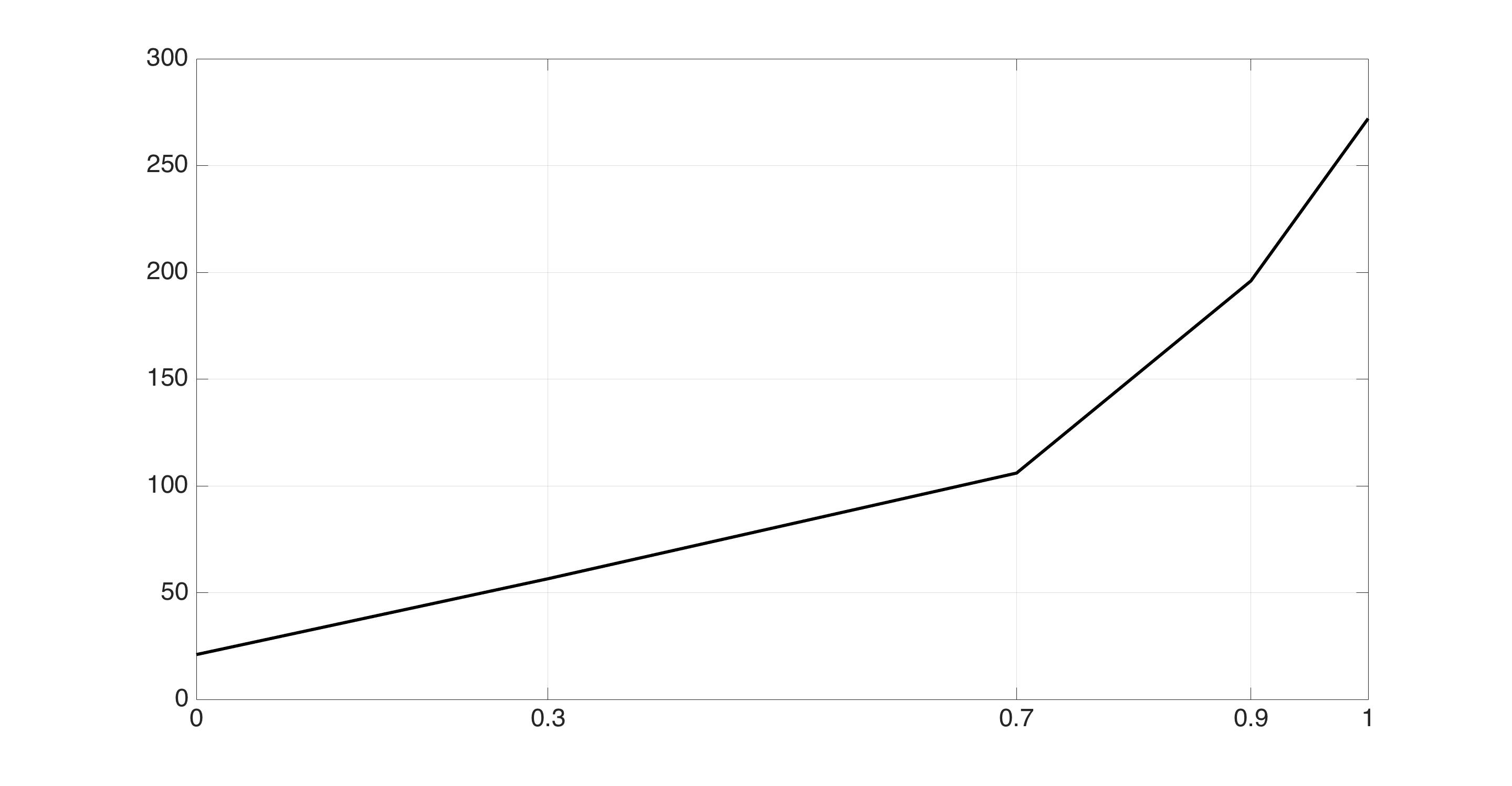}
   \caption{Histogram and respective quantile function of  $H_{9E}$  in Example \ref{ex1}.}\label{figHist1}
\end{figure}

For the interval-valued variable Arrival Delay for airline 9E, $I_{9E}=[-68,744],$ the quantile function is the linear function 
$\Psi_{9E}^{-1}(t)= 338+(2t-1)406, \, \, t\in[0,1]$,
represented in Figure \ref{figInt1}.

\begin{figure}[h!] 
  \centering
   \includegraphics[width=0.5\textwidth]{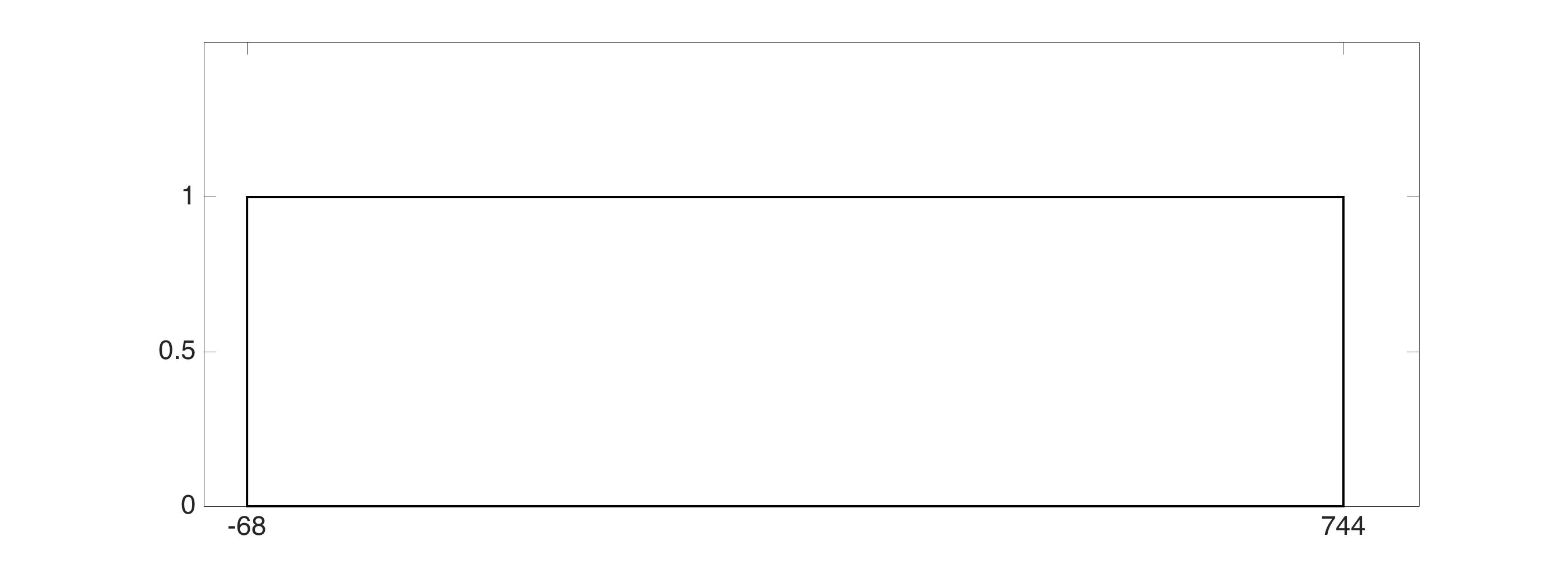}~
    \includegraphics[width=0.5\textwidth]{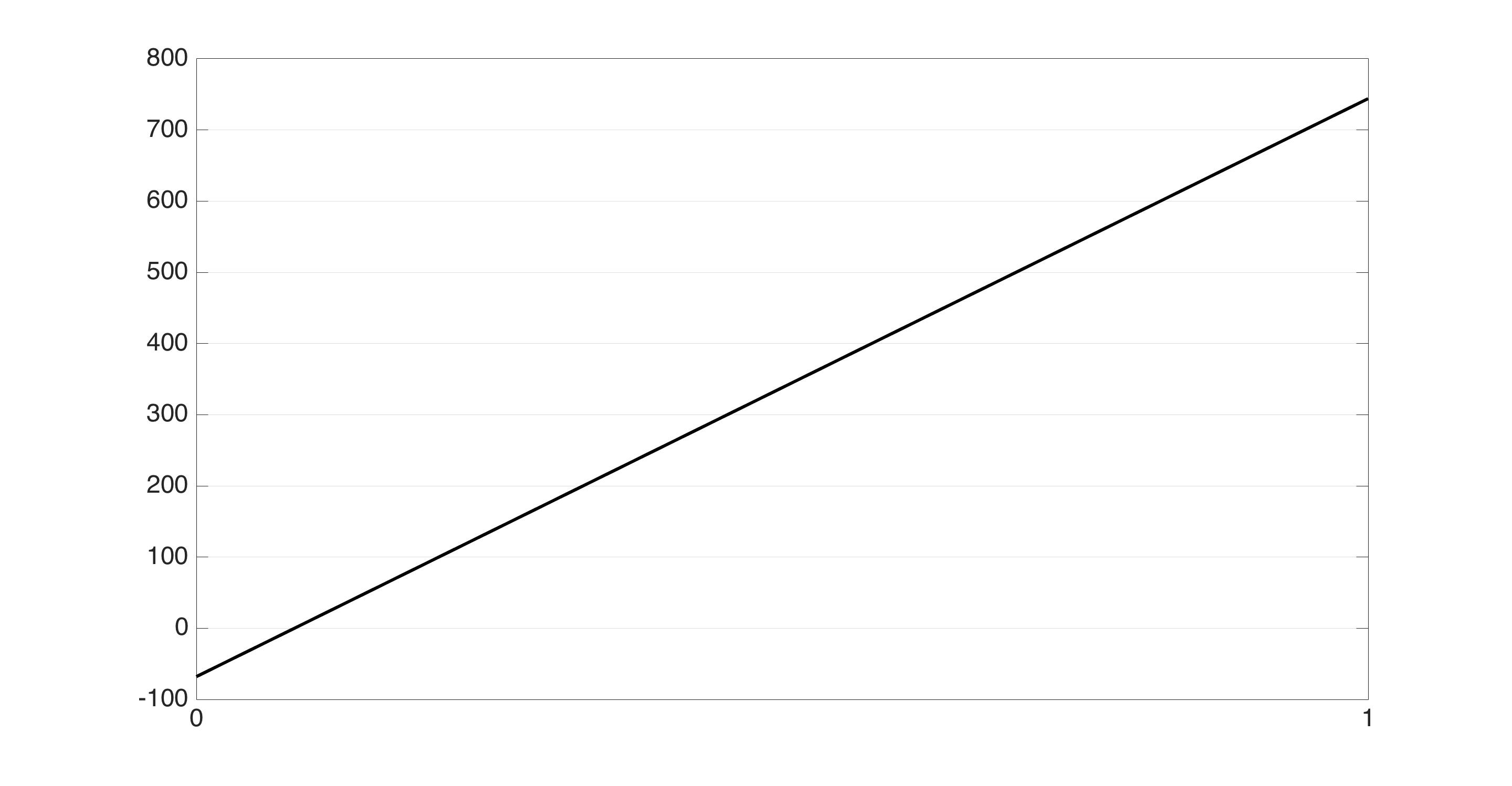}
   \caption{Interval and quantile function of $I_{9E}$  in Example \ref{ex1}.}\label{figInt1}
\end{figure}

\end{example}

Since empirical quantile functions are the inverse of cumulative distribution functions, which  under the uniformity hypothesis are piecewise linear functions in the case of the histograms and continuous linear functions in the case of the intervals, with domain $\left[0,1\right],$ we shall use the usual arithmetic operations with functions. However, when we use quantile functions as the representation of histograms some issues may arise: 

\begin{enumerate}

\item  To operate with the quantile functions it is convenient to define all functions involved with the same number of pieces, and the domain of each piece must be the same for all functions. All histograms must hence have the same number of subintervals and the weight associated with each corresponding subinterval must be equal in all units \citep{irve06};

\item The quantile functions are always non-decreasing functions because all linear pieces represent the subintervals $ [\underline{I}_{X(j)_i} , \overline{I}_{X(j)_i}];$

\item When we multiply a quantile function by a negative real number we obtain a function that is not a non-decreasing function, so it cannot representd a histogram/interval;

\item The quantile function corresponding to the symmetric of the histogram $H_X$ is the quantile function $-\Psi_{X}^{-1}(1-t)$ with $t \in [0,1]$ and not the function obtained by multiplying $\Psi_{X}^{-1}(t)$ by $-1.$ As it is required for quantile functions, $-\Psi_{X}^{-1}(1-t)$ is a non-decreasing function;

\item  $\Psi_{X}^{-1}(t)-\Psi_{X}^{-1}(1-t)$ is not a null function, as might be expected, but a quantile function with null (symbolic) mean \citep{bidi03};

\item The functions $-\Psi_{X}^{-1}(1-t)$ and $\Psi_{X}^{-1}(t)$ are linearly independent, providing that $-\Psi_{X}^{-1}(1-t) \neq \Psi_{X}^{-1}(t)$; only when the histogram $H_{X}$ is symmetric with respect to the $yy-$axis we have $-\Psi_{X}^{-1}(1-t) = \Psi_{X}^{-1}(t).$ 
\end{enumerate}

For more details about the behavior of quantile functions in the representation of histograms, see \citet{tesedias14, dibr15}.

Descriptive statistics for symbolic variables have been proposed by several authors, see \cite{bertrand2000}  for interval-valued variables and  \cite{bidi03,bidi07}, for histogram-valued variables. The definition of symbolic mean, which is needed in the sequel, is as follows.

\begin{definition}\label{defmean}

Let $Y$ be a histogram-valued variable and $H_Y(i)$ the observed histogram for each unit $i \in \left\{1,...,n\right\}$, composed by $m_{i}$ subintervals with weight $p_{i\ell},$, $\ell \in\left\{1,\ldots,m_{i}\right\}$.
The symbolic mean of \textit{histogram-valued variable} $Y$  is given by \citep{bidi03}
$\overline{Y}  =  \displaystyle \frac{1}{n} \sum_{i=1}^{n} \sum_{\ell=1}^{m_{i}} c_{Y(i)\ell} \,\, p_{i\ell}.$
For an interval-valued variable $Y$, the symbolic mean is the arithmetic mean of the centres of the intervals (\cite{bertrand2000}):
$\overline{Y}  =  \displaystyle \frac{1}{n}\sum_{i=1}^{n} c_{Y(i)}.$
\end{definition}

\subsection{Mallows distance}\label{s2.2}

In recent literature the Mallows distance is considered as an adequate measure to evaluate the similarity between distributions. This distance has been successfully used in cluster analysis for histogram data \citep{irve06}, in forecasting histogram time series \citep{arma09}, and in linear regression with histogram/interval-valued variables \citep{arma09,irve15, dibr15,dibr17}.

To compare the distributions taken by the histogram-valued variables using the Mallows distance, they should be represented by their corresponding quantile functions. This distance is defined as follows:

\begin{definition}\label{defDM}
Given two quantile functions $\Psi_{X}^{-1}(t)$ and $\Psi_{Y}^{-1}(t)$ that represent the distributions of the histogram-valued variables $X$ and $Y,$ the Mallows distance   \citep{mall72} is defined as:
\begin{equation}\label{eqDMdef}
D_M(\Psi_{X}^{-1}(t),\Psi_{Y}^{-1}(t))=\sqrt{\int_{0}^{1}(\Psi_{X}^{-1}(t)-\Psi_{Y}^{-1}(t))^2dt}
\end{equation}

Assuming the Uniform distribution within the subintervals and that the quantile functions $\Psi_{X}^{-1}(t)$ and $\Psi_{Y}^{-1}(t)$ are both written with $m$ pieces and the same set of cumulative weights, the squared Mallows distance may be rewritten as \citep{irve06}:  
\begin{equation}\label{eqDM}
D^{2}_M(\Psi_{X}^{-1}(t),\Psi_{Y}^{-1}(t))=\displaystyle\sum_{\ell=1}^{m}p_{\ell}\left[(c_{X_\ell}-c_{Y_{\ell}})^2+\frac{1}{3}(r_{X_\ell}-r_{Y_{\ell}})^2\right]
\end{equation}

\noindent where $c_{X_\ell}, c_{Y_\ell}$ are the centres and $r_{X_\ell}, r_{Y_\ell}$ are the half-ranges of subinterval $\ell$ of variables $X$ and $Y,$ respectively, with $\ell \in \left\{1,2,\ldots,m \right\}.$

\end{definition}

We notice that the weight of the difference between the centres is larger than the weight of the difference between the half-ranges.

Given a set of $n$ units, we may then compute the \textit{Mallows barycentric histogram} or simply \textit{barycentric histogram}, $Y_{b},$  represented by the quantile function $\Psi_{Y_{b}}^{-1}(t)$, as the solution of the minimization problem \citep{irve06}:
\begin{equation}
{\renewcommand{\arraystretch}{1.5}
\begin{array}{ll}
\min &  {\displaystyle \sum_{i=1}^{n}D_{M}^2(\Psi_{Y(i)}^{-1}(t),\Psi_{Y_{b}}^{-1}(t))}.\\
\end{array}}\label{eq1.Barycentric}
 \end{equation}

According to (\ref{eqDM}) we may rewrite the previous problem in the form
\begin{equation}
\begin{array}{ll}
\min &   {\displaystyle f(c_{b1},r_{b1},\ldots,c_{bm},r_{bm})=\displaystyle \sum_{i=1}^{n}\sum_{\ell=1}^{m}p_{i}\left[\left(c_{Y(i)_{\ell}}-c_{Y_{b\ell}}\right)^2+\frac{1}{3}\left(r_{Y(i)_{\ell}}-r_{Y_{b\ell}}\right)^2\right]}.\\
\end{array}
\label{eq2.Barycentric}
 \end{equation}
 The optimal solution is obtained by solving a least squares problem, and is a histogram where the centre and half range of each subinterval $\ell$ is the classical mean, respectively, of the centres and of the half ranges $\ell,$ of all units $i$, and it corresponds to the quantile function:
\begin{equation}\label{eq.BarycentricQF}
\Psi_{Y_{b}}^{-1}(t)=\left\{{\renewcommand{\arraystretch}{1.25}\begin{array}{lll}
                   c_{b1}+\left(\frac{2t}{w_{1}}-1\right) r_{b1} & \textit{ if } & 0 \leq t < w_{1} \\
                    c_{b2} +\left(\frac{2(t-w_{1})}{w_{2}-w_{1}}-1\right) r_{b2} & \textit{ if } & w_{1} \leq t < w_{2} \\
                     \vdots & & \\
                     c_{bm} +\left(\frac{2(t-w_{m-1})}{1-w_{m-1}}-1\right) r_{bm} & \textit{ if } & w_{m-1} \leq t \leq
                   1
 \end{array}}
  \right..
\end{equation}

\begin{equation*}
\text{with  } c_{b\ell}=\frac{1}{n}\displaystyle \sum_{i=1}^{n} c_{Y(i)_{\ell}} \qquad \text{and} \qquad r_{b\ell}=\frac{1}{n}\displaystyle \sum_{i=1}^{n} r_{Y(i)_{\ell}}.
\end{equation*}

%

\begin{proposition}\label{prBarycentricMeanFQ}
	\citep{irve10} The quantile function $\Psi_{Y_b}^{-1}(t),$ that represents the barycentric histogram of $n$ histograms, is the mean of the $n$ quantile functions that represent each observation of the histogram-valued variable $Y,$ i.e.
	$\Psi_{Y_b}^{-1}(t)=\overline{\Psi_{Y}^{-1}}(t).$
\end{proposition}

The mean value of the barycentric histogram $\overline{Y}_{b},$ is the symbolic mean of the histogram-valued variable $Y$ \citep{tesedias14}:
$\displaystyle \overline{Y}=\int_{0}^{1} \overline{\Psi_{Y}^{-1}}(t) dt=\int_{0}^{1} \Psi_{Y_b}^{-1}(t) dt.$

The barycentric histogram corresponds to the ``center of gravity'' of the set of histograms. This may be observed in the following example.

\begin{example} \label{ex3}

Consider the distributions of variable Air Time for the airlines in Example \ref{ex1}, the barycentric histogram may be represented by the quantile function
\small{
$$
\overline{\Psi_{AirTime}^{-1}}(t)=\left\{\begin{array}{lll}
                     43.26+\frac{t}{0.2}\times 12.06  & \quad if \quad  & 0 \leq t <0.2 \\
                    61.56 +\frac{t-0.2}{0.2} \times 6.24 &  \quad if \quad & 0.2 \leq t < 0.4\\
                   77 +\frac{t-0.4}{0.2} \times 9.2 &  \quad if \quad & 0.4 \leq t < 0.6 \\
                    95.88 +\frac{t-0.6}{0.2} \times 9.68 &  \quad if \quad & 0.6 \leq t < 0.8 \\
                   162.08 +\frac{t-0.8}{0.2} \times 56.52 & \quad if \quad & 0.8 \leq t \leq 1
 \end{array}
  \right.
  $$}

The quantile functions that represent the distributions for the five airlines, and the respective barycentric histogram, are represented in Figure \ref{figBarHist}. This Barycenter provides more information than the symbolic mean (see Definition \ref{defmean}), that is equal to $87.96.$

\begin{figure}[h!] 
  \centering
    \includegraphics[width=1\textwidth]{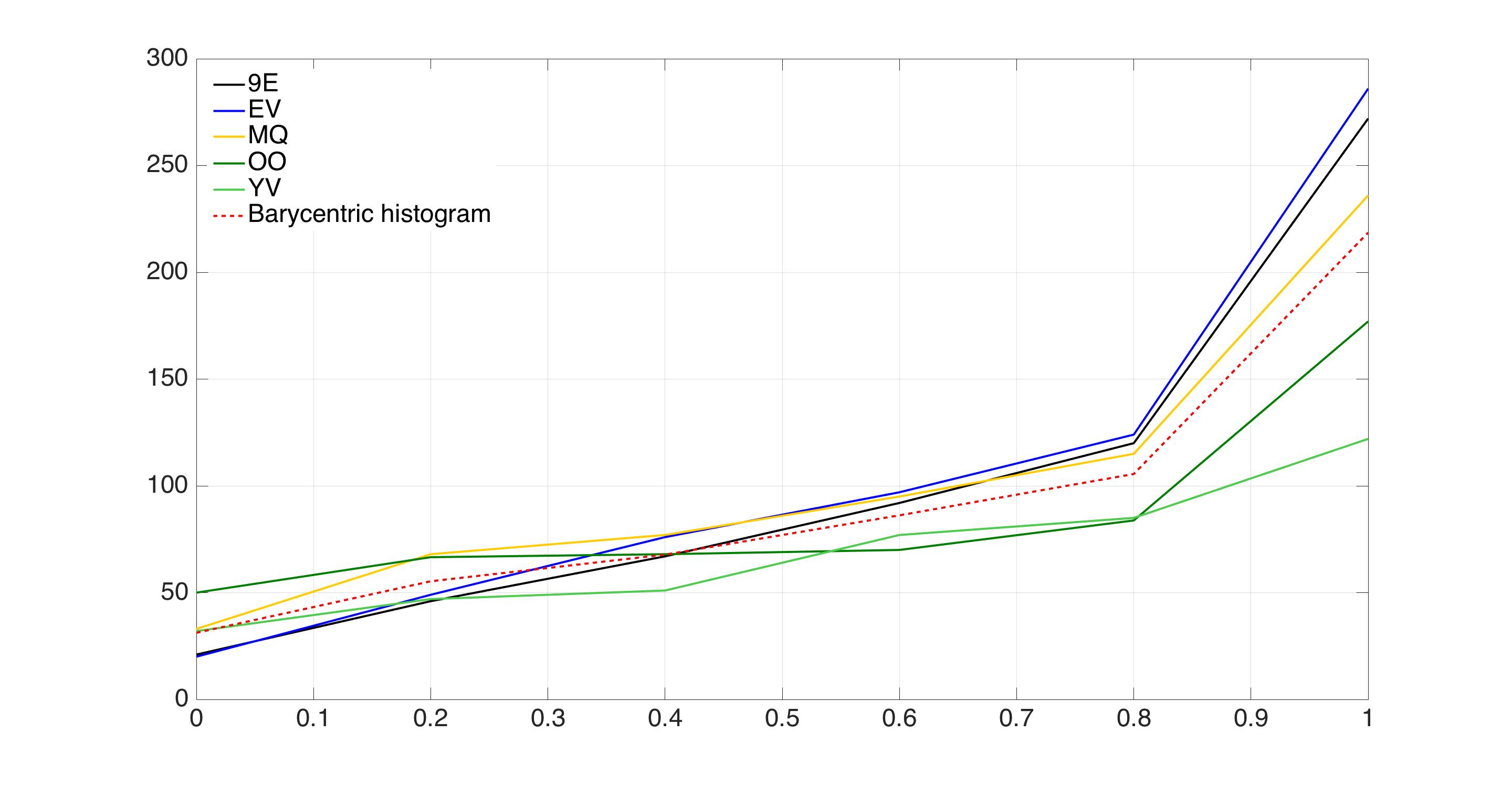}
   \caption{Observed and barycentric quantile functions of variable Air Time, in Example \ref{ex1}.}\label{figBarHist}
\end{figure}


\end{example}



The concept of barycentric histogram allows for the definition of a measure of inertia based on the \textit{Mallows distance} \cite{irve06}.
The total inertia $(TI),$ with respect to the barycentric histogram   $Y_{b},$ of a set of $n$ histogram observations $Y(i), i=1,\ldots, n$, is given by
$TI= \displaystyle \sum_{i=1}^{n} D^{2}_{M}(\Psi_{Y(i)^{-1}}(t),\Psi_{Y_b^{-1}}(t)).$
\cite{irve06} proved that the \textit{Mallows distance} allows for the Huygens decomposition of inertia for clustered histogram-valued  data:
\begin{equation}
\begin{array}{lll}
 TI & = & BI+WI \\
 & = & \displaystyle \sum_{k=1}^{s}n_k  D^{2}_{M}(\Psi_{Y_{b_k}^{-1}}(t),\Psi_{{Y_b}^{-1}}(t))+ \displaystyle \sum_{k=1}^{s} \sum_{i=1}^{n_k} D^{2}_{M}(\Psi_{Y(i)^{-1}}(t),\Psi_{Y_{b_k}^{-1}}(t))
 \end{array}\label{eq.HuygensTheorem}
\end{equation}

\noindent where $Y_{b_{k}}$ is the barycenter of Group $k$ and $n_{k} = |G_{k}|$ $k\in \left\{1,\ldots,s\right\}.$

\subsection{Linear combination of histograms}\label{s2.3}

The linear combination of histogram-valued variables is not a simple adaptation of the classical definition, as  it is not possible to apply the classical linear combination to the quantile functions that represent the distributions, as in:
$\Psi_{\widehat{Y}(i)}^{-1}(t)=a_{1}\Psi_{X_{1}(i)}^{-1}(t)+a_{2}\Psi_{X_{2}(i)}^{-1}(t)+\ldots+a_{p}\Psi_{X_{p}(i)}^{-1}(t).$
The problem comes from the fact that when we multiply a quantile function by a negative number we do not obtain a non-decreasing function. If non-negativity constraints are imposed on the parameters $a_j, $ $j \in \{1,2, \ldots, p\}$ a quantile function is always obtained; however, this solution forces a direct linear relation between $\Psi_{\widehat{Y}(i)}^{-1}(t)$ and $\Psi_{X_{j}(i)}^{-1}(t)$, which is not acceptable.

In order to define a linear combination that solves the problem of the semi-linearity of the space of the quantile functions and allows for a direct or an inverse linear relation between the involved histogram-valued variables, a method was proposed in \cite{dibr15,dibr17}.  The proposed definition includes two terms for each explicative variable - one for the quantile function that represents each histogram/interval $X_{j}(i)$ and the other for the quantile function that represents the respective symmetric histogram/interval. This solution, however, increases the number of parameters to estimate, one should bear in mind that the number of units $n$ should be large enough. 

\begin{definition}\label{defDSD} \citep{dibr15}
Consider the histogram-valued variables $X_{1}; X_{2}; \ldots; X_{p}$.
Let $\Psi_{X_1(j)}^{-1}(t),$ $\Psi_{X_2(i)}^{-1}(t),\ldots,$ $\Psi_{X_p(i)}^{-1}(t)$,with $t \in [0,1]$, be the quantile functions that represent the distributions the variables take for each unit $i$ and $-\Psi_{X_1(i)}^{-1}(1-t),-\Psi_{X_2(i)}^{-1}(1-t),\ldots,$ $-\Psi_{X_p(i)}^{-1}(1-t),$  the quantile functions that represent the respective symmetric distributions.
The linear combination of  the histogram-valued variables $X_{1}, X_{2}, \ldots, X_{p}$ is a new histogram-valued variable $Y$, each quantile function $\Psi_{Y(i)}^{-1}$ may be expressed as
\begin{equation}\label{eqPredictFQint1}
\Psi_{Y(i)}^{-1}(t)=\sum_{j=1}^{p}a_{j}\Psi_{X_{j}(i)}^{-1}(t)-\sum_{j=1}^{p}b_{j}\Psi_{X_{j}(i)}^{-1}(1-t)
\end{equation}

\noindent with $t \in \left[0,1\right];$ $a_{j},b_{j} \geq 0,$  $j \in \left\{1,2,\ldots,p \right\}.$

\end{definition}

In the particular case of interval-valued variables and assuming uniformity within the observed intervals, the corresponding quantile functions are   $\Psi^{-1}_{X_{j}(i)}(t)=c_{X_{j}(i)}+(2t-1)r_{X_{j}(i)}$ and  $-\Psi^{-1}_{X_{j}(i)}(1-t)=-c_{X_{j}(i)}+(2t-1)r_{X_{j}(i)}, i=1,\ldots, n.$
In this case, the linear combination may be written as \citep{dibr17}:
\begin{equation}\label{eqPredictFQint2}
\Psi_{Y(j)}^{-1}(t)=\sum_{i=1}^{p}\left(a_{i} -b_{i}\right)c_{X_{i}(j)}+\sum_{i=1}^{p}\left(a_{i}+b_{i}\right) r_{X_{i}(j)} \left(2t-1\right)
\end{equation}
\noindent with $t \in \left[0,1\right];$ $a_{j},b_{j} \geq 0,$  $j \in \left\{1,2,\ldots,p \right\}.$

Interval-valued variables where all observations are degenerate intervals, i.e., intervals with null range, are classical variables. In this case (\ref{eqPredictFQint2}) is the classical definition of linear combination. This is in accordance with the purpose of SDA, that the statistical concepts and methods defined for symbolic variables should generalize the classical ones.

\section{Linear Discriminant function}\label{s3}

The process that allows obtaining a linear discriminant function for histogram-valued variables is similar to the classical method.

\subsection{Discriminant function for classical variables}\label{s3.1}

In the classical setup, i.e., with real-valued variables, the linear discriminant function defines a score for each unit, $S$, as  the linear combination of the $p$ descriptive variables $\xi=X' \gamma$ where $X$ is the vector of $p$ centered variables and $\gamma$ the vector of weights $p\times 1.$ The sum of squares of the resulting discriminant scores is given by $\xi'\xi=\gamma' T \gamma,$ where $T=XX'$  is the matrix of the total Sums of Squares and Cross-Products (SSCP) of the matrix $X$. Note that $\gamma' T\gamma$ is the sum of the  squared Euclidean distances $d$ between $S(i)=\displaystyle\sum_{j=1}^{p}\gamma_{j}X_j(i)$ and $\overline{S}=\displaystyle \frac{1}{n}\displaystyle\sum_{i=1}^{n}S(i)$, i.e. $
\displaystyle \gamma'T\gamma=\sum_{i=1}^{n}d^2(S(i),\overline{S}). 
$
Given the decomposition of $T$  as the sum of the matrix of the sum of squares and cross-products between-groups, $B$ and the matrix of the sum of squares and
cross-products within-groups, $W,$  $T = B +W$, we may write $\gamma' T \gamma=\gamma' B \gamma+\gamma' W \gamma.$
The vector $\gamma$ that defines the discriminant function is then estimated such that the ratio between the variability between groups and the variability within groups is maximum, i.e., maximizing $ \displaystyle \lambda=\frac{\gamma' B\gamma}{\gamma' W\gamma}$. This is an easy and classical optimization problem.

\subsection{Discriminant function for histogram-valued variables}\label{s3.2}

We now define the discriminant function that allows obtaining a discriminant score for $p$ histogram-valued variables, as well as the ratio to be optimized to estimate the weight vector of the discriminant function.

\begin{definition} \label{deffuncaodiscr}
Consider $p$  histogram-valued variables represented for each unit $i$ by the respective quantile function $\Psi^{-1}_{X_{j}(i)}(t)$, as in (\ref{eqHFQ}). The score for unit $i$ is   a quantile function, $\Psi_{S(i)}^{-1}(t),$ obtained by the linear combination as in (\ref{eqPredictFQint1}):
\begin{equation}\label{eqScore1}
\Psi_{S(i)}^{-1}(t)=\sum_{j=1}^{p}a_{j}\Psi_{X_{j}(i)}^{-1}(t)-
\sum_{j=1}^{p}b_{j}\Psi_{X_{j}(i)}^{-1}(1-t)
\end{equation}
\noindent with $t \in \left[0,1\right];$ $a_{j},b_{j} \geq 0,$  $j \in \left\{1,2,\ldots,p \right\}.$

\end{definition}

Similarly to the classical case, the sum of the squared Mallows distance between the score for unit $i$, $\Psi_{S(i)}^{-1}(t),$ and the mean of all scores - the quantile function $\overline{\Psi_{S}^{-1}}(t)$ may be written as $\gamma' T\gamma$, where $\gamma=(a_1, \ldots, a_p, b_1, \ldots, b_p)$ is the weight vector and $T$ the matrix of the total Sums of Squares and Cross-Products (SSCP) for the $p$ histogram-valued variables.

For the next results we shall use the following notation:

\begin{itemize}
\item  $\Psi_{S(i)}^{-1}(t)$: quantile function representing the score obtained applying (\ref{eqScore1}), in Definition \ref{deffuncaodiscr}, to the quantile functions $\Psi^{-1}_{X_{j}(i)}(t),$ defined in (\ref{eqHFQ}). For each subinterval $\ell$ this function is defined by 
\begin{equation}\label{eqScore}
\begin{array}{l}
\displaystyle  \sum_{j=1}^{p} \left(a_j c_{X_j(i)_\ell}-b_j c_{X_j(i)_{(m-\ell+1)}}\right)+ \\
\displaystyle + \left(\frac{2(t-w_{\ell-1})}{w_{\ell}-w_{\ell-1}}-1\right)\sum_{j=1}^{p} \left(a_j r_{X_j(i)_\ell}+b_j r_{X_j(i)_{(m-\ell+1)}}\right).
\end{array}
\end{equation}

\item   $\overline{\Psi_{S}^{-1}}(t)$ - $\mathit{barycentric \, score}$: quantile function that represents the barycentric histogram of the $n$ scores, it is the mean of the quantile functions that represent the individual scores. For subinterval $\ell,$  $\overline{\Psi_{S}^{-1}}(t)$ is defined by 
\begin{equation}\label{eqScoreMedia}
\begin{array}{l}
\displaystyle  \sum_{j=1}^{p} \left(a_j \overline{c}_{X_{j_\ell}}-b_j \overline{c}_{X_{j_{(m-\ell+1)}}}\right)+\\
\displaystyle 
+ 
\left(\frac{2(t-w_{\ell-1})}{w_{\ell}-w_{\ell-1}}-1\right)\sum_{j=1}^{p} \left(a_j \overline{r}_{X_{j_\ell}}+b_j \overline{r}_{X_{j_{(m-\ell+1)}}}\right)$$
\end{array}
\end{equation}
\noindent where $\overline{c}_{X_{j_\ell}}$ and $\overline{r}_{X_{j_\ell}}$ are, respectively,  the means of the centres and half ranges of the subinterval $\ell$ for variable $j.$

\item  $\overline{\Psi_{S_k}^{-1}}(t)$ - $\mathit{barycentric \, group \, score}$: quantile function that represents the barycentric histogram of the scores in group $k, (k=1,\ldots, s)$ i.e., it is the mean of the quantile functions that represent the scores in group $k.$ For each subinterval $\ell$, $\overline{\Psi_{S_k}^{-1}}(t)$  is defined by 
\begin{equation}\label{eqScoreMediaGrupo}
\begin{array}{l}
\displaystyle  \sum_{j=1}^{p} \left(a_j \overline{c}_{X_{jk\ell}}-b_j \overline{c}_{X_{jk{(m-\ell+1)}}}\right)+ \\
\displaystyle
+ \left(\frac{2(t-w_{\ell-1})}{w_{\ell}-w_{\ell-1}}-1\right)\sum_{j=1}^{p} \left(a_j \overline{r}_{X_{jk\ell}}+b_j \overline{r}_{X_{jk{(m-\ell+1)}}}\right)$$
\end{array}
\end{equation}
\noindent where $\overline{c}_{X_{jk\ell}}$ and $\overline{r}_{X_{jk\ell}}$ are the means for the observations in group $k,$ of the centers and half ranges of  subinterval $\ell$ for variable $j.$ 

\end{itemize}

\begin{proposition} \label{propT}
Let  $\Psi_{S(i)}^{-1}(t)$ be the score (quantile function) obtained from Definition (\ref{deffuncaodiscr}) considering $p$ histogram-valued variables $\Psi^{-1}_{X_{j}(i)}(t),$ $j\in\left\{1,\ldots,p \right\},$  
and $\overline{\Psi_{S}^{-1}}(t),$ the barycentric score. The sum of the squared  Mallows distance between $\Psi_{S(i)}^{-1}(t)$ and $\overline{\Psi_{S}^{-1}}(t)$ may be written as 
$\displaystyle \sum_{i=1}^{n}D_M^2(\Psi_{S(i)}^{-1}(t),\overline{\Psi_{S}^{-1}}(t))=\gamma' T \gamma,$ 
 where $T$ is the matrix of the total Sums of Squares and Cross-Products (SSCP) of the $p$ histogram-valued variables and $\gamma$ is the $p\times 1$ vector of weights.  

\end{proposition} 

\begin{proof}

Consider the quantile functions $\overline{\Psi_{S(i)}^{-1}}(t)$  and $\overline{\Psi_{S}^{-1}}(t)$ defined as in (\ref{eqScore}) and (\ref{eqScoreMedia}), respectively.  Applying (\ref{eqDM}) in  Definition \ref{defDM} for the Mallows distance and after some algebra, it is possible to write $\displaystyle \sum_{i=1}^{n}D_M^2(\Psi_{S(i)}^{-1}(t),\overline{\Psi_{S}^{-1}}(t)),$ in  matricial form as $\gamma' T \gamma$ where,

\begin{itemize}

\item $\gamma$ is the  $p\times 1$ vector of non-negative weights, i.e. $\gamma'=\begin{array}{ccccc}[a_1 & b_1  & \ldots & a_p & b_p] \end{array},$ with $a_j,b_j \geqslant 0,$ $\forall j \in \left\{1, \ldots, p \right\}.$

\item $T$ is the matrix of the total SSCP  obtained by the product $A \times A',$ where $A'$ is a $2mn\times 2p$ matrix defined as:

$$A'=\left[\begin{array}{cccc}
\mathbf{A_{11}} & \mathbf{A_{12} }& \ldots & \mathbf{A_{1p}} \\
\mathbf{A_{21}} & \mathbf{A_{22}} & \ldots & \mathbf{A_{2p} }\\
\cdots & \cdots & \ldots & \cdots \\
\mathbf{A_{m1}} & \mathbf{A_{m2}} & \ldots & \mathbf{A_{mp} 
}\end{array}
\right]$$

where $\mathbf{A_{\ell j}}$,  for $\ell \in \left\{1,2,\ldots,m \right\}$ and $j \in \left\{1,2,\ldots,p \right\}$, are  $2n \times 2$ matrices. The elements in each matrix $\mathbf{A_{\ell j}}=[a_{hq}],$ with $h\in \left\{1,2,\ldots,2n \right\}$ and $q \in \left\{1,2\right\}$  are defined as:
$$
a_{hq}=\left\{
\begin{array}{lll}
                  \sqrt{p_\ell}\left(c_{X_{j}(h)_{\ell}} - \overline{c}_{X_{j_\ell}} \right)  & if & 1\leqslant h \leqslant n  \textit{ and } q=1 \\
- \sqrt{p_\ell}\left(c_{X_{j}(h)_{m-\ell+1}}- \overline{c}_{X_{j_{(m-\ell+1)}} } \right)  & if & n+1\leqslant h \leqslant 2n  \textit{ and } q=1 \\
 \sqrt{\frac{p_\ell}{3}}\left(r_{X_{j}(h)_{\ell}}- \overline{r}_{X_{j_\ell}} \right)  & if & 1\leqslant h \leqslant n  \textit{ and } q=2 \\
- \sqrt{\frac{p_\ell}{3}}\left(r_{X_{j}(h)_{m-\ell+1}} - \overline{r}_{X_{j_{(m-\ell+1)}} } \right)  & if & n+1\leqslant h \leqslant 2n  \textit{ and } q=2 \\
 \end{array} 
 \right.
$$

 $T= A \times A'$ is a symmetric matrix of order $2p$, its elements  $t_{hq},$  $h,q\in \left\{1,2,\ldots,2p \right\}$ are defined as:
{\scriptsize
$$
t_{hq}=\left\{\begin{array}{lll}
                   \displaystyle\sum\limits_{i=1}^n \displaystyle\sum\limits_{\ell=1}^m p_{\ell}\left(\tilde{c}_{X_{\frac{h+1}{2}}(i)_{\ell}}
                   \tilde{c}_{X_{\frac{q+1}{2}}(i)_{\ell}}+\frac{1}{3}\tilde{r}_{X_{\frac{h+1}{2}}(i)_{\ell}}\tilde{r}_{X_{\frac{q+1}{2}}(i)_{\ell}} \right)  & if & \textit{ h,q are odd }  \\
                        \displaystyle\sum\limits_{i=1}^n \displaystyle\sum\limits_{\ell=1}^m p_{\ell}\left(\tilde{c}_{X_{\frac{h}{2}}(i)_{(m-\ell+1)}}
                   \tilde{c}_{X_{\frac{q}{2}}(i)_{(m-\ell+1)}}+\frac{1}{3}\tilde{r}_{X_{\frac{h}{2}}(i)_{(m-\ell+1)}}\tilde{r}_{X_{\frac{q}{2}}(i)_{(m-\ell+1)}} \right)  & if & \textit{ h,q are even}  \\
                         \displaystyle\sum\limits_{i=1}^n \displaystyle\sum\limits_{\ell=1}^m p_{\ell}\left(-\tilde{c}_{X_{\frac{h}{2}}(i)_{\ell}}
                   \tilde{c}_{X_{\frac{q+1}{2}}(i)_{(m-\ell+1)}}+\frac{1}{3}\tilde{r}_{X_{\frac{h}{2}}(i)_{\ell}}\tilde{r}_{X_{\frac{q+1}{2}}(i)_{(m-\ell+1)}} \right) & if & \textit{h is even, q is odd } \\
 \end{array} \right.
$$
}

\normalsize
with $\tilde{c}_{X_{\frac{\delta+1}{2}}(i) \ell}=c_{X_{\frac{\delta+1}{2}}(i) \ell}-
\overline{c}_{X_{\frac{\delta+1}{2}} \ell}$ and $\tilde{r}_{X_{\frac{\delta+1}{2}}(i) \ell}=r_{X_{\frac{\delta+1}{2}}(i) \ell}-
\overline{r}_{X_{\frac{\delta+1}{2}}\ell},$ $\delta \in \{h,q\}.$
\end{itemize}
\end{proof}

According to the Huygens theorem (\ref{eq.HuygensTheorem}), the following decomposition holds:
\begin{equation}\label{eqHuygnesDecomp}
\begin{array}{l}
 \sum\limits_{i=1}^{n} D^{2}_{M}(\Psi_{S(i)}^{-1}(t),\overline{\Psi_{S}^{-1}}(t))=\\
  \quad =\sum\limits_{k=1}^{s}|G_{k}| D^{2}_{M}(\overline{\Psi_{S}^{-1}}(t),\overline{\Psi_{S_{k}}^{-1}}(t))+\sum\limits_{k=1}^{s}\sum\limits_{i \in G_{k}} D^{2}_{M}(\Psi_{S(i)}^{-1}(t),\overline{\Psi_{S_{k}}^{-1}}(t))
  \end{array}
  \end{equation}

\noindent with $|G_{k}|$ the cardinal of group $k, $ and the quantile functions $\Psi_{S(i)}^{-1}(t)$ - $\mathit{score}$, $\overline{\Psi_{S}^{-1}}(t)$ - $\mathit{barycentric \, score} $ and $\overline{\Psi_{S_k}^{-1}}(t)$ - $\mathit{barycentric \, group \, score}$, as defined in (\ref{eqScore}),  (\ref{eqScoreMedia}) and (\ref{eqScoreMediaGrupo}), respectively. 

This result shows that, under our framework, the SSCP may also be decomposed
 in the sum of the squares and crossproducts between groups and the sum of the
  squares and crossproducts within groups. We may then write
$\gamma' T \gamma= \gamma' B \gamma+\gamma' W \gamma,$
where $B=[b_{hq}]$ and $W=[w_{hq}]$ are symmetric matrices of order $2p$.
The elements $b_{hq}$ of the matrix $B,$ with $h,q \in \{1,\ldots,2p \}$ are defined as:

{\scriptsize
$$
b_{hq}=\left\{\begin{array}{lll}
                   \displaystyle\sum\limits_{k=1}^s n_k \displaystyle\sum\limits_{\ell=1}^m p_{\ell}\left(\tilde{\overline{c}}_{X_{\frac{h+1}{2}k\ell}}
                   \tilde{\overline{c}}_{X_{\frac{q+1}{2} k\ell}}+\frac{1}{3}\tilde{\overline{r}}_{X_{\frac{h+1}{2}k\ell}}\tilde{\overline{r}}_{X_{\frac{q+1}{2}k\ell}} \right)  & if & \textit{ h,q are odd }  \\
                        \displaystyle\sum\limits_{k=1}^s n_k \displaystyle\sum\limits_{\ell=1}^m p_{\ell}\left(\tilde{\overline{c}}_{X_{\frac{h}{2}k(m-\ell+1)}}
                   \tilde{\overline{r}}_{X_{\frac{q}{2}k(m-\ell+1)}}+\frac{1}{3}\tilde{\overline{r}}_{X_{\frac{h}{2}k(m-\ell+1)}}\tilde{\overline{r}}_{X_{\frac{q}{2}k(m-\ell+1)}} \right)  & if & \textit{ h,q are even}  \\
                         \displaystyle\sum\limits_{k=1}^s n_k \displaystyle\sum\limits_{\ell=1}^m p_{\ell}\left(-\tilde{\overline{c}}_{X_{\frac{h}{2}k\ell}}
                   \tilde{\overline{c}}_{X_{\frac{q+1}{2}k(m-\ell+1)}}+\frac{1}{3}\tilde{\overline{r}}_{X_{\frac{h}{2}k\ell}}\tilde{\overline{r}}_{X_{\frac{q+1}{2}k(m-\ell+1)}} \right) & if & \textit{h is even, q is odd } \\
 \end{array} \right.
$$
}
\normalsize
\noindent with $\tilde{\overline{c}}_{X_{\frac{\delta+1}{2}k\ell}}=\overline{c}_{X_{\frac{\delta+1}{2}\ell}}-\overline{c}_{X_{\frac{\delta+1}{2}k\ell}}$ and $\tilde{\overline{r}}_{X_{\frac{\delta+1}{2}k\ell}}=\overline{r}_{X_{\frac{\delta+1}{2}\ell}}-\overline{r}_{X_{\frac{\delta+1}{2}k\ell}},$ $\delta \in \left\{ h,q \right\}.$

The elements $w_{hq}$ of matrix $W$,  $h,q \in \{1,\ldots,2p\}$ are defined as:
{\scriptsize
$$
w_{hq}=\left\{\begin{array}{lll}
                   \displaystyle\sum\limits_{k=1}^s  \displaystyle\sum\limits_{i=1}^{n_k}\displaystyle\sum\limits_{\ell=1}^m p_{\ell}\left(\tilde{\overline{c}}_{X_{\frac{h+1}{2}k}(i)_\ell}
                   \tilde{\overline{c}}_{X_{\frac{q+1}{2}k}(i)_\ell}+\frac{1}{3}\tilde{\overline{r}}_{X_{\frac{h+1}{2}k}(i)_\ell}\tilde{\overline{r}}_{X_{\frac{q+1}{2}k}(i)_\ell} \right)  & if & \textit{ h,q are odd }  \\
                    \displaystyle\sum\limits_{k=1}^s  \displaystyle\sum\limits_{i=1}^{n_k}\displaystyle\sum\limits_{\ell=1}^m p_{\ell}\left(\tilde{\overline{c}}_{X_{\frac{h}{2}k}(i)_{(m-\ell+1)}}
                   \tilde{\overline{r}}_{X_{\frac{q}{2}k}(i)_{(m-\ell+1)}}+\frac{1}{3}\tilde{\overline{r}}_{X_{\frac{h}{2}k}(i)_{(m-\ell+1)}}\tilde{\overline{r}}_{X_{\frac{q}{2}k}(i)_{(m-\ell+1)}} \right)  & if & \textit{ h,q are even}  \\
                          \displaystyle\sum\limits_{k=1}^s  \displaystyle\sum\limits_{i=1}^{n_k}\displaystyle\sum\limits_{\ell=1}^m p_{\ell}\left(-\tilde{\overline{c}}_{X_{\frac{h}{2}k}(i)_\ell}
                   \tilde{\overline{c}}_{X_{\frac{q+1}{2}k}(i)_{(m-\ell+1)}}+\frac{1}{3}\tilde{\overline{r}}_{X_{\frac{h}{2}k}(i)_\ell}\tilde{\overline{r}}_{X_{\frac{q+1}{2}k}(i)_{(m-\ell+1)}} \right) & if & \textit{h is even, q is odd } \\
 \end{array} \right.
$$
}
\normalsize
\noindent with $\tilde{\overline{c}}_{X_{\frac{\delta+1}{2}k}(i)_\ell}$ and $\tilde{\overline{r}}_{X_{\frac{\delta+1}{2}k}(i)_\ell}$ as above.

As in the classic case, the parameters of the model, i.e., the components of  vector $\gamma$,  are estimated such that the ratio of the variability between groups and the variability within groups is maximized. The optimization problem is now written as
\begin{equation}\label{CFQP}
\mbox{Maximize }\lambda=\frac{\gamma' B\gamma}{\gamma' W\gamma} \,\, \mbox{ subject to } \, \gamma \geqslant 0.
\end{equation}

Contrary to the classical situation, this is a hard optimization problem as it is now necessary to solve a constrained fractional quadratic problem. The methods to solve this problem are presented in the next section.

\subsection{Optimization of constrained fractional quadratic problem }\label{s3.3}
Problem (\ref{CFQP}) is non-convex, and finding the global optima in this class requires a computational effort that increases exponentially with the size, in this case, of matrices $B$ and $W$. 
\\
Exact methods for global optima, as Branch and Bound, are heavy in terms of memory and computational time. Attempting to solve the  instances of problem (\ref{CFQP}), from our data, using proven software like Baron \citep{Baron}, has failed even for small size problems. Typically, good feasible solutions were found in the first iterations but the algorithm ended up unable to establish the global optimality of the best solution found, even when largely increasing the maximum number of iterations. The incumbent solution  obtained by the software, $\tilde{\gamma}$, allowed to define a lower bound for the optimal solution: 
\begin{equation}
\tilde{\lambda}=\frac{\tilde{\gamma}' B\tilde{\gamma}}{\tilde{\gamma}' W\tilde{\gamma}} \le \lambda^*= \max_{{\gamma} \geqslant 0} \frac{\gamma' B\gamma}{\gamma' W\gamma}.
\end{equation}

To improve this lower bound, or to prove optimality, it is necessary to introduce an upper bound:
$\tilde{\lambda} \le \lambda^* \le \bar{\lambda}.$
If for a small $ \epsilon$, we have $ \bar{\lambda}-\tilde{\lambda} \le \epsilon $ then we accept $\tilde{\gamma}$ as an  $ \epsilon-\mbox{optimal solution}$. In a numerical method we do not expect to have a zero gap solution ($\epsilon=0$), so we defined  $\epsilon=10^{-4}$ as the intended accuracy for our numerical results. 

Tight bounds for general constrained fractional quadratic problems based on copositive relaxation were proposed in \cite{Amaral2014}. 
Let $X\bullet Y= trace (XY)$ be the Frobenius inner product of two matrices $X$ and $Y$ in the set of symmetric matrices of size $n$, $\mathcal{M}_n$. The cone of completely positive matrices is given by
$\mathcal{C}_n^*=\left\{ X \in \mathcal{M}_n : X=YY' , Y \mbox{ an $n \times k$ matrix with }Y\ge O \right\}.$ Since
\begin{equation}
 \lambda^*= \max_{{\gamma} \geqslant 0} \frac{\gamma' B\gamma}{\gamma' W\gamma}
      = \max_{{\gamma} \geqslant 0\; \gamma' W\gamma=1} \gamma' B\gamma,    
\end{equation}
taking $\Gamma=\gamma' \gamma \in \mathcal{C}_n^*$ with rank$(\Gamma)=1$, and considering that $\gamma' W\gamma={W}\bullet \Gamma$, 
following \cite{Amaral2014} we obtain the following completely positive reformulation of (\ref{CFQP}):
\begin{equation}\label{cop}
\max \left\{ {B}\bullet\Gamma : {W}\bullet \Gamma=1, \Gamma \in \mathcal{C}^*_{n}\right\}\, .\end{equation}
Checking condition $\Gamma \in
\mathcal{C}^*_{n}$ is (co-)NP-hard \cite{MuKa87,Dick11b}, but 
knowing that the cone of symmetric, non-negative and semi-definite matrices, denoted as doubly non-negative matrices, represented by  
$\mathcal{D}_n=\left\{ X \in \mathcal{M}_n : y^TXy \ge 0, \forall y \in \mathbb{R}^n\; ,  X\ge O \right\}\,$
provides an approximation for $\mathcal{C}_n^*$, since $\mathcal{C}_n^*\subseteq  \mathcal{D}_n$, it is then possible to exploit this relaxation to obtain an upper bound for
(\ref{cop}),  by solving
\begin{equation}\label{dnn}
\bar{\lambda}=\max \left\{ {B}\bullet\Gamma : {W}\bullet \Gamma=1, \Gamma \in \mathcal{D}_{n}\right\}\, .\end{equation}
This upper bound was enough in all instances to close the gap for $\epsilon=10^{-4}$ and to prove $\epsilon-$optimality of the incumbent solution provided by Baron. This allowed using this solution with confidence to estimate the model parameters.

\subsection{Classification in two \textit{a priori} groups}\label{s3.4}

From the discriminant function in Definition \ref{deffuncaodiscr}  and using the Mallows distance, it is possible to classify an unit in one of two groups, $G_{1},$ $G_{2}$.
Let $\overline{\Psi^{-1}_{S_{G_1}}}(t)$ and $\overline{\Psi^{-1}_{S_{G_2}}}(t)$,  be the quantile functions that represent the barycentric score of each group and let $\Psi^{-1}_{S(i)}(t)$ be  the quantile function that represents the score of  unit $i.$ 
Unit $i$ is assigned to Group $G_1$ if
$D^{2}_M\left(\Psi_{S(i)}^{-1}(t),\overline{\Psi^{-1}_{S_{G_1}}}(t)\right)<D^{2}_M\left(\Psi_{S(i)}^{-1}(t),\overline{\Psi^{-1}_{S_{G_2}}}(t)\right),$
otherwise it is assigned to Group $G_2$; i.e.,
it is assigned to the group for which the Mallows distance between its score and the corresponding barycentric score is minimum.

\section{Experiments with simulated data} \label{s4}

In this section, we evaluate the performance of the proposed discriminant method, for  histogram and interval-valued variables, under different conditions.

\subsection{Description of the simulation study} \label{s4.1}

Symbolic data tables that illustrate different situations were created; a full factorial design was employed, similar for histogram and interval-valued variables, with the following factors:
	
	\begin{itemize}
  	\item Three variables: $X_1, X_2, X_3$ and two groups: $G_1,G_2$
	\item Four learning sets and two test sets: 
		\begin{description}
			\item[Learning sets:] $|G_1|=10,$ $|G_2|=40$;   $|G_1|=|G_2|= 25$ ; 
			$|G_1|=50,$ $|G_2|=200$;  $|G_1|=|G_2|=125$.
			\item[Test Sets:] $|G_1|=200,$ $|G_2|=800$;  $|G_1|=|G_2|=500$.
		\end{description}	
		
	\item Different levels of similarity between the groups 
	\begin{itemize}
	\item For histogram-valued variables - four levels defined by the mean and standard deviation of the distributions:
		\begin{description}
			\item[ Case HA:] Similar mean and similar standard deviation;
			\item[ Case HB:] Similar mean and different standard deviation;
			\item[ Case HC:] Different mean and similar standard deviation;
			\item[ Case HD:] Different mean and different standard deviation.
		\end{description}	
		
	For each case above, four different distributions are considered: Uniform, Normal,Log-normal, Mixture of distributions.
	\end{itemize}
	\begin{itemize}
	\item For interval-valued variables - six levels defined by the centers and half-ranges of the intervals:
		\begin{description}
			\item[ Case IA:] Low variation in half range and similar center;
			\item[ Case IB:] Large variation in half range and similar center;
			\item[ Case IC:] Low variation in centre and similar half range;
			\item[ Case ID:] Large variation in centre and similar half range;
			\item[ Case IE:] Low variation in half range and center;
			\item[ Case IF:] Large variation in half range and center.
		\end{description}	

	\end{itemize}

\end{itemize}

To simulate symbolic data tables for the conditions described above, it is necessary to generate the observations of the variables $X_j$ in the two groups. For the case of histogram-valued variables, different distributions are considered, whereas for the case of the interval-valued variables several types of half ranges and centers are considered. 

\textit{Histogram-valued variables}

\begin{enumerate} 

\item For each variable, the values of the mean and the standard deviation were fixed. For the three variables in this study, we selected: $\mu_{X_1}=20, \sigma_{X_1}=8;$ $\mu_{X_2}=10, \sigma_{X_2}=6;$ $\mu_{X_3}=5, \sigma_{X_3}=4.$

\item For each variable $j$ and for each group $k$, two vectors of length $n$ are generated, one with values of the means, $M_{X_{jk}}=[m_{jk}(i)]$ and another with values of the standard deviation  $S_{X_{jk}}=[s_{jk}(i)].$ The $n$ values of each vector $M_{X_{jk}}$ and $S_{X_{jk}},$ are randomly generated, as follows:
	\begin{itemize}
	\item $m_{jk}(i)\sim \mathcal{U}(c_1(1+a),c_2(1+a)),$ with $c_1=0.6 \times \mu_{X_j},$ $c_2=1.4 \times \mu_{X_j},$ $a=0$ in Group $G_1$ and $a>0$ in Group $G_2.$ ($a=0.1$ - cases \textbf{HA, HB}; $a=0.5$ - cases \textbf{HC, HD}).
	\item $s_{jk}(i)\sim \mathcal{U}(h_1(1+b),h_2(1+b)),$ with $h_1=0.6 \times \sigma_{X_j},$ $h_2=1.4 \times \sigma_{X_j},$ $b=0$ in Group $G_1$ and $b>0$ in Group $G_2.$   ($b=0.1$ - cases \textbf{HA, HC}; $b=0.5$ - cases \textbf{HB, HD}). 
	\end{itemize}

\item  From each couple of values $m_{jk}(i)$ and $s_{jk}(i),$ $i\in \left\{1,\ldots,n\right\}$ we randomly generate 5000 real values, $x_{jki}(w),$ $w \in \left\{1,\ldots,5000\right\} $ that allow creating the histograms corresponding to unit $i$  and variable $X_j$ of the group $k$ and distribution $D.$
 According to the distribution $D,$ the real values are generated as follows:
	   \begin{itemize}
            \item D=Uniform distribution:  $x_{jki}(w) \sim \mathcal{U}(a_{jk}(i),b_{jk}(i))$ with $a_{jk}(i)=m_{jk}(i)-\sqrt{3}s_{jk}(i)$ and $b_{jk}(i)=m_{jk}(i)+\sqrt{3}s_{jk}(i)$ 
            
             \item D=Normal distribution: $x_{jki}(w) \sim \mathcal{N}(m_{jk}(i),s_{jk}(i))$
                         
            \item  D=Log-Normal distribution: $x_{jki}(w)\sim Log\mathcal{LN}(\tilde{m}_{jk}(i),\tilde{s}_{jk}(i))$ with $\tilde{m}_{jk}(i)=\frac{1}{2} \ln  \left(\frac{(m_{jk}(i))^4}{(s_{jk}(i))^2+(m_{jk}(i))^2}\right)$ and $\tilde{s}_{jk}(i)= \ln \left(\frac{(s_{jk}(i)))^2}{(m_{jk}(i))^2}\right)$
                \end{itemize}
                
 \item  The 5000 real values $x_{jki}(w)$ generated for each unit $i$ are organized in histograms, which represent its empirical distribution. For all units, all subintervals of each histogram have the same weight, $0.20.$     
 \end{enumerate} 
 
\textit{Interval-valued variables}

For each variable, the values of the center and the half range were fixed, as follows: $c_{X_1}=20, r_{X_1}=8;$ $c_{X_2}=10, r_{X_2}=6;$ $c_{X_3}=5, r_{X_3}=4.$

 \begin{enumerate} 
 \item For each variable $j$, in each group $k$, two vectors of length $n$ are randomly generated, one with values of the centers, $C_{X_{jk}}=[c_{jk}(i)]$ and another with values of the half ranges $R_{X_{jk}}=[r_{jk}(i)],$ as follows:
	\begin{itemize}
	\item $c_{jk}(i)\sim \mathcal{U}(c_1(1+a),c_2(1+a)),$ with $c_1=0.6 \times c_{X_j},$ $c_2= 1.4 \times c_{X_j},$ $a=0$ in Group $G_1$ and $a>0$ in Group $G_2$ ($a=0.05$ - cases \textbf{IA, IB}; $a=0.2$ - cases \textbf{IC, IE} and $a=0.6$ - cases \textbf{ID, IF}).

	\item $r_{jk}(i)\sim \mathcal{U}(h_1(1+b),h_2(1+b)),$ with $h_1=0.6 \times r_{X_j},$ $h_2=1.4 \times r_{X_j},$ $b=0$ in Group $G_1$ and $b>0$ in Group $G_2$  ($b=0.05$ - cases \textbf{IC, ID}; $b=0.2$ - cases \textbf{IA, IE} and $b=0.6$ - cases \textbf{IB, IF}).
 
	\end{itemize}

\item  From each couple of values $c_{jk}(i)$ and $r_{jk}(i),$ $i\in \left\{1,\ldots,n\right\}$  the intervals associated with each unit $i$ are generated: $[c_{jk}(i)-r_{jk}(i);c_{jk}(i)+r_{jk}(i)[.$

 \end{enumerate}  
	
For each case, 100 data tables were generated for each learning set, one data table for each test set was created under the same conditions.
For each replicate of each of the four learning sets, the classification rule based on the proposed discriminant method was applied and the proportion of well assigned units calculated. The mean values and corresponding standard deviations of the obtained results are provided as Supplementary Material (Tables \ref{table1} and \ref{table3}). 

The discriminant functions with the parameters obtained (see Definition \ref{deffuncaodiscr}) for each replicate of the learning sets with the same/different number of units, were then applied to the test set in which the groups have the same/different size. The proposed classification rule (see Subsection \ref{s3.4}) was applied and the proportion of well assigned cases in each group calculated; the obtained values are presented in Supplementary Material (Tables \ref{table2}, \ref{table4}). 

\subsection{Discussion of results} \label{s4.2}

\textit{Histogram-valued variables}

Tables \ref{table1} and \ref{table2} (in Supplementary Material), gather the mean and standard deviation of the proportion of well classified units in test and learning sets, for each of the four cases, under different conditions; Figures \ref{fig1} and \ref{fig2} represent those values for the test sets.

\begin{figure}[h!] 
  \centering
    \includegraphics[width=1\textwidth]{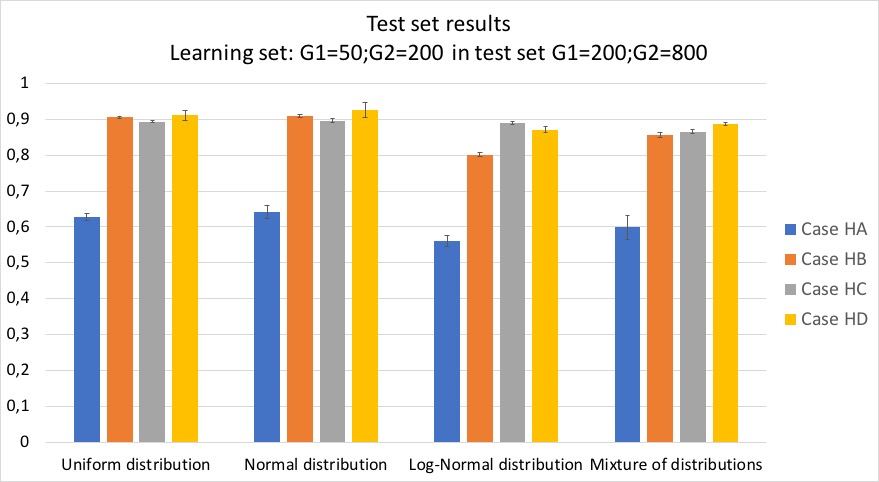}
      \caption{Mean and standard deviation of the proportion of well classified units for test set with $|G1|=200$ and $|G2|= 800$.}\label{fig1}
\end{figure}

\begin{figure}[h!] 
  \centering
    \includegraphics[width=1\textwidth]{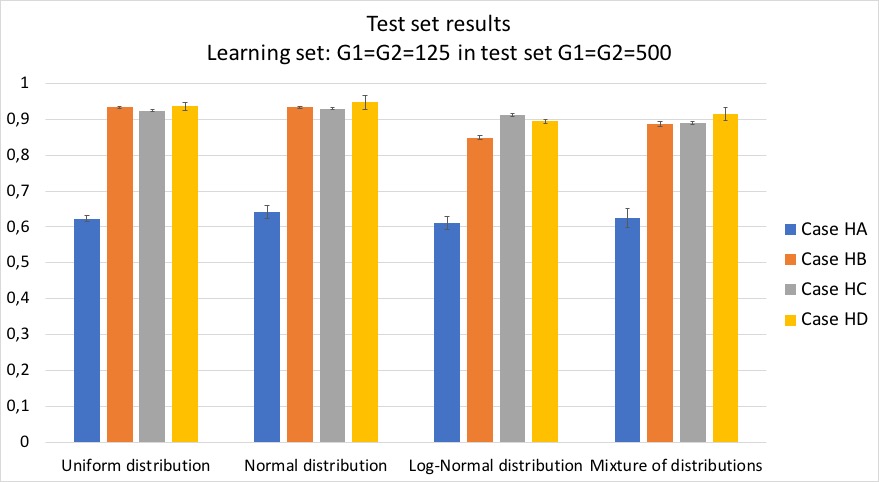}
      \caption{Mean and standard deviation of the proportion of well classified units for test set with $\protect |G1|=|G2|= 500$.} \label{fig2}
\end{figure}

 We  observe similar behaviors for the different distributions in both learning and test sets. In general, increasing the difference between means and/or standard deviations of the two groups provides a better discrimination. The mean of the proportion of well classified is generally slightly higher in the situations where the classes are balanced.
In almost every situation, the mean hit rate is influenced in the same way by the mean and standard deviation; in the case of the Log-normal distribution, when the mean of the two groups is quite different, the increase of the perturbation in the standard deviation seems to have little influence.
In the learning sets, the mean of the proportion of well classified units is slightly higher for smaller samples. This behavior is observed mainly for the Uniform and Normal distributions. 

\vspace{0.5 cm}

\textit{Interval-valued variables}

Tables \ref{table3} and  \ref{table4} (in Supplementary Material) gather the mean and standard deviation of the proportion of well classified units in the learning and  test sets, for each of the cases, under different conditions. Figure \ref{fig3} illustrates the behavior in the test sets, for all situations investigated.

\begin{figure}[h!] 
  \centering
    \includegraphics[width=1\textwidth]{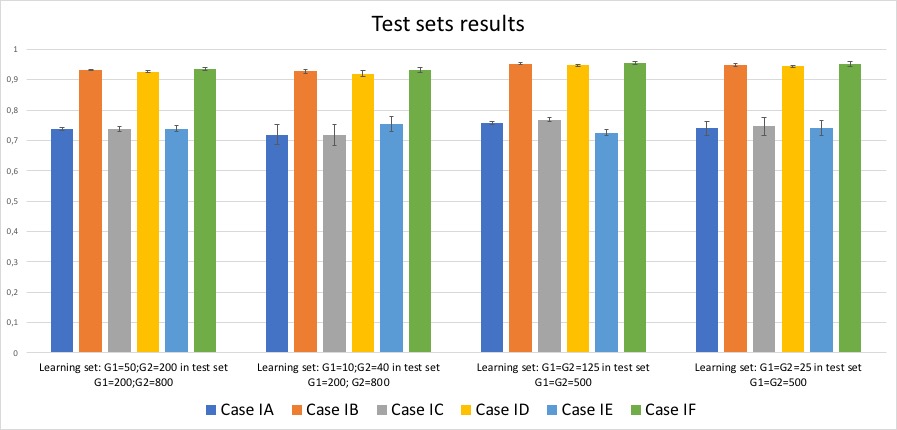}
      \caption{Mean and standard deviation of the proportion of well classified units for test sets defined for interval-valued variables.}\label{fig3}
\end{figure}

The behavior observed is similar for all considered cases. In general, no differences are observed when the distinction between groups is induced by the centers or half-ranges. When the groups are balanced, the proportion of well classified units is slightly higher.  As in the case of histograms, in the learning sets, the mean hit rate is slightly higher for smaller samples.
Observing Figure \ref{fig3}, we notice that when the learning set is smaller, the standard deviations in the corresponding test sets are larger. This behavior is observed mainly in smaller samples of cases IA, IC and IE.

In conclusion, the classification method based on the proposed discriminant function behaves as expected in all considered situations, for both interval and histogram-valued variables, providing good results (in terms of hit rate) in a wide variety of situations.

\section{Application - Flights that Departed NYC in 2013} \label{s5}

This study concerns all outgoing flights from the three New York City airports (LGA, JFK, and EWR) in 2013. The microdata is available in the R package  \textit{nycflights13}, and was originally obtained from the website of the Bureau of Transportation Statistics \footnote{\url{https://www.transtats.bts.gov/DL\_SelectFields.asp},
	accessed 2020-03-10.}. 
We considered the Flights Data that include information about date of departure,  departure and arrival delays, airlines, and air time. The original data contains information concerning flights of the $15$ airlines departing from NYC in 2013, in a total of $327346$ records.

\begin{table}[h!]
\caption{Original 'microdata' (partial view).}
\begin{center} 
\renewcommand{\arraystretch}{1.1}
\renewcommand{\tabcolsep}{0.2cm}
{\scriptsize
\begin{tabular}{c|c|cccc}
 \hline
Airline &  Date  &  $DDELAY $ & $ADELAY$ &  AIRTIME \ldots \\
  \hline
  \multirow {4} {*} {AA} & 2013,1,1 & $-4$  &  $-2$ & 377 & \ldots \\
   &  2013,1,1 & $-1$  & $11$ & 358 &\ldots\\
 & $\vdots$ & $\vdots$ & $\vdots$ &   $\vdots$ &   $\vdots$  \\
   &  2013,1,31 & $-7$ & $-30$ & $359$ & \ldots \\
  \hline
   \multirow {4} {*} {9E} & 2013,1,1 & $16$  &  $4$ &  $171$& \ldots\\
   & 2013,1,4 & $2$  & $10$ & $116$ & \ldots\\
 & $\vdots$ & $\vdots$ & $\vdots$ & $\vdots$ & $\vdots$ \\
   &  2013,1, 31 & $13$ & $-12$ & $118$ & \ldots\\
  \hline
 & $\vdots$ & $\vdots$ & $\vdots$  & $\vdots$  & $\vdots$\\
\end{tabular}}\label{tableOriginal}
\end{center}
\end{table}
The goal is to discriminate Regional and Main carrier, from the flights' departure and arrival delays (DDELAY and ADELAY, respectively, negative times represent early departures or arrivals) and the Airtime, (all recorded in minutes). Part of the original data is illustrated in Table \ref{tableOriginal}. However, the units under study are not the flights but the airlines. Since the amount of information associated with each airline is substantial, we build a symbolic data table, where each unit is the airline/month. We consider only the units where variability was observed, thereby excluding three airlines as well as some months for other specific airlines (where values were constant for some variable(s)), leading to a final symbolic data array with 141 units. The included airlines are (IATA Codes): 9E; AA; B6; DL; EV; FL; MQ;UA; US; VX; WN; YV, four of which are Regional Carriers (9E; EV; MQ; YV), the remaining eight are Main Carriers. In spite of the original number of observations of each variable for each airline/month not being equal, we considered, for all units, histograms where the subintervals have the same weight, 0.20. Part of the symbolic data array is represented in Table \ref{tablesymbolic}.

\begin{table}[h!]
\caption{Symbolic histogram-valued data array (partial view).}
\begin{center}
\renewcommand{\arraystretch}{1.1}
\renewcommand{\tabcolsep}{0.2cm}
{\scriptsize
\begin{tabular}{c|c}
 \hline
Airline / Month &  AIRTIME \\
  \hline
AA /Jan &   $\{[30,139[,0.2;[139,158[,0.2;[158,205[,0.2;[205,255[,0.2; [255,408],0.2 \}$  \\
 $\vdots$ & $\vdots$   \\
AA/Dec & $\{[30,138[,0.2;[138,159[,0.2;[159,203[,0.2;[203,297.8[,0.2; [297.8,426],0.2 \}$ \\
  \hline
 9E/Jan &   $\{[24,43[,0.2;[43,58[,0.2;[58,86[,0.2;[86,121.4[,0.2; [121.,264],0.2 \}$  \\
 $\vdots$ & $\vdots$   \\
9E/Dec & $\{[26,54[,0.2;[54,76[,0.2;[76,104[,0.2;[104,134[,0.2; [134,261],0.2 \}$ \\
  \hline
 $\vdots$ & $\vdots$
\end{tabular}}\label{tablesymbolic}
\end{center}
\end{table}
Alternatively, we may consider interval-valued variables, and register for each unit (airline /month) the range of values recorded for the corresponding flights. The resulting data array is represented in Table \ref{table_interval}.

\begin{table}[h!]
	\caption{Symbolic interval-valued data array (partial view).}
	\begin{center}
		\renewcommand{\arraystretch}{1.1}
		\renewcommand{\tabcolsep}{0.2cm}
		{\scriptsize
			\begin{tabular}{c|c|c|c}
				\hline
				Airline / Month & DDELAY& ADELAY & AIRTIME \\
				\hline
				AA /Jan & $[-16,337]$& $[-54,368]$ &  $[30,408] $  \\
				$\vdots$ & $\vdots$ & $\vdots$ & $\vdots$   \\
				AA/Dec & $[-16,896]$& $[-51,878]$ &   $[30,426] $ \\
				\hline
				9E/Jan &  $[-18,360]$& $[-59,370]$ &   $[24,264] $  \\
				$\vdots$ & $$& $$ &   $\vdots$   \\
				9E/Dec & $[-19,360]$& $[-50,386]$ &  $[26,261]$ \\
				\hline
				$\vdots$ & $\vdots$ & $\vdots$ & $\vdots$
		\end{tabular}}\label{table_interval}
	\end{center}
\end{table}

The methodology described in Sections \ref{s3.2} and \ref{s3.3} lead to the following discriminant function, which then allows obtaining the discriminant score of each unit
$ \Psi^{-1}_{S(i)}(t) =-0.7844\Psi^{-1}_{AIRTIME(i)}(1-t).$
Considering interval-valued variables, the discriminant scores are obtained from
 $\Psi^{-1}_{S(i)}(t) =0.0322\Psi^{-1}_{DDELAY(i)}(t)+0.9597\Psi^{-1}_{AIRTIME(i)}(t).$

As described in Section \ref{s3.4}, each unit is assigned to the group for which the Mallows distance between its score and the corresponding barycentric score is minimum. Figure  \ref{figApril} illustrates the barycentric scores of the two groups and the scores of the units of all airlines in April, for the histogram and the interval data analysis. When histogram-valued variables are used, FL and US are misclassified in Group 2, whereas with interval-valued variables, only FL is misclassified. The proportion of well classified units was the same with and without Leave-One-Out, results are summarized in Table \ref{tableresults}. We may observe that the hit rate is higher when interval-valued variables are considered, although in this case less information is used about the variables. This may be explained by the different behavior of the parameters used in the determination of the discriminant scores.
\begin{figure}[h!]
	\centering
	\begin{tabular}[b]{c}
		\includegraphics[width=.43\linewidth]{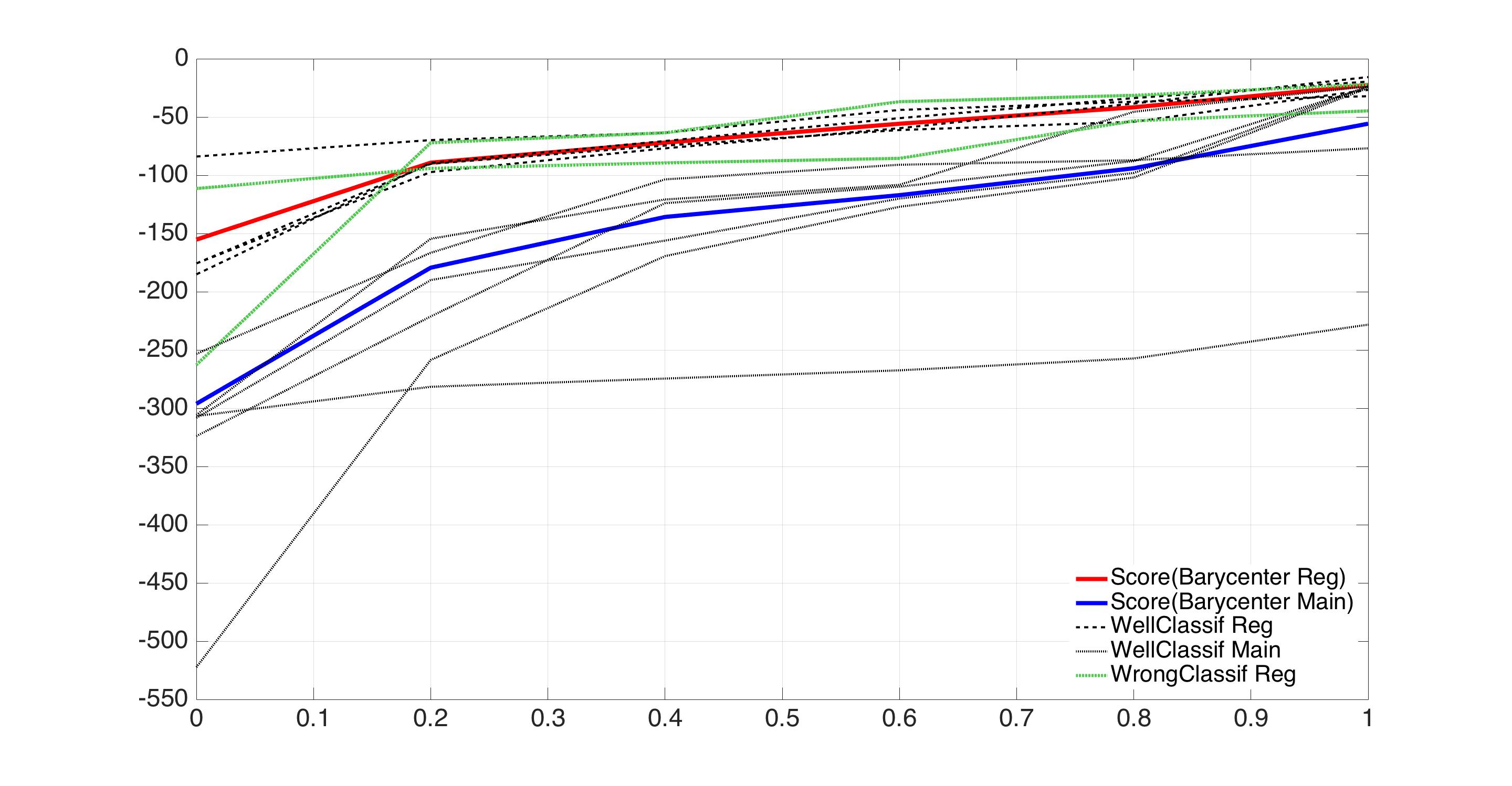} \\
		\small (a)
	\end{tabular} \qquad
	\begin{tabular}[b]{c}
		\includegraphics[width=.43\linewidth]{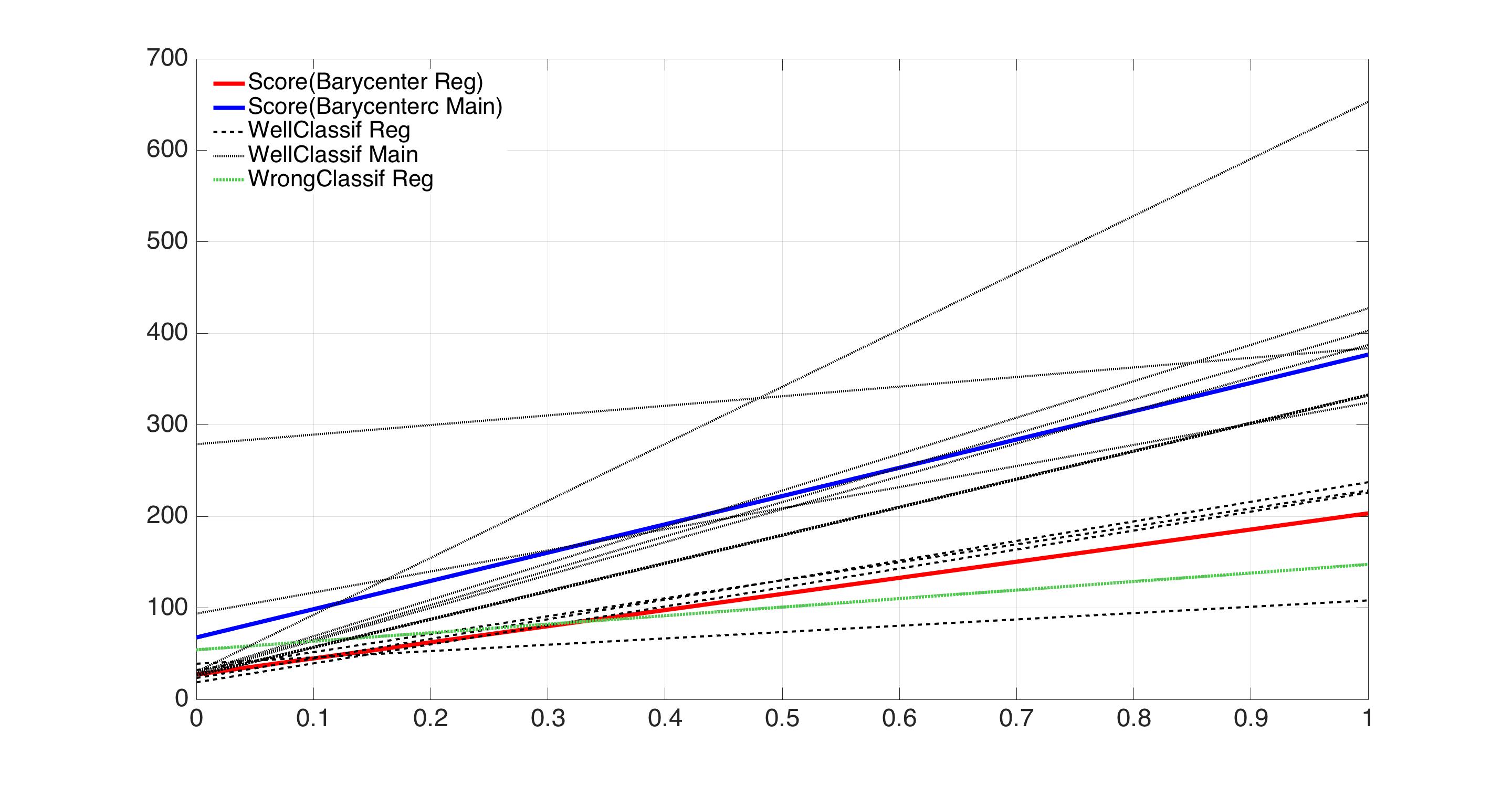} \\
		\small (b)
	\end{tabular}
	\caption{Barycentric scores of the two groups and scores of the units for all airlines  in April. (a) Histogram data (b) Interval data.}\label{figApril}
\end{figure}

\begin{table}
\caption{Classification results considering histogram and interval-valued variables with and without Leave-One-Out.}\label{tableresults}
\begin{center}
{\scriptsize
\begin{tabular}{p{0.2\textwidth} | p{0.3\textwidth}  p{0.4\textwidth}}
\hline
 Type of variable & \multicolumn{2}{c}{Classification with and without LOO}  \\
 \hline
\multirow{3}{*}{Histogram-variables} &  $\%$ Well classified & $83\%$ \\ 
\cline{2-3}
                                      & Units wrongly classified in G1  & None \\ 
                                      &  Units wrongly classified in G2  &  24 units: airlines FL and US, all months \\
                                      \hline
\multirow{4}{*}{Interval-variables} &  $\%$ Well classified & $90\%$  \\
\cline{2-3}
                                      & Units wrongly classified in G1  & None \\ 
                                      &  Units wrongly classified in G2  &   14 units: all months of airline FL \\
                                      & & and months May and Sept of US\\ 
                                      \hline
                                     \end{tabular}}
\end{center}
\end{table}

%

\section{Conclusion} \label{s6}

The discriminant method for histogram-valued variables proposed in this paper allows defining a score for each unit, in the form of a quantile function and of the same type as the variables' records - i.e. a histogram when we have histogram-valued variables and an interval when we have interval-valued variables. The score is obtained by an appopriate  linear combination of variables, where now the model parameters are obtained by the optimization of a constrained fractional quadratic problem. This is a non-trivial problem, since it is non convex, and is solved by using  Branch and Bound and Conic Optimization. The obtained scores allow for a classification in two a priori groups, based on the Mallows distance between the score of each unit and the barycentric score of each class. The methodology is defined for histogram-valued variables, but also applies to interval-valued variables, as an interval is a special case of a histogram; for degenerate intervals, i.e. real values, we obtain the classic method.
The proposed method performs well in a diversity of situations.  It is potentially useful in a wide variety of areas where variability inherent to the data is relevant for the classification task.

Future research perspectives include developing the method to allow for the classification in more than two groups.

\section*{Acknowledgments}

This work is financed by National Funds through the Portuguese funding agency, FCT - Funda\c{c}\~{a}o para a Ci\^{e}ncia e a Tecnologia, within projects UIDB/50014/2020 and UIDB/00297/2020.

\section*{References}

\bibliography{XBib_EJOR}

\begin{landscape}

\section*{Supplementary Material} \label{s7}

\begin{table}[h!]
\begin{center}
{\scriptsize
\begin{tabular}{c|l|cccc}
  \hline
  &  Learning set & Case HA  &   Case HB &   Case HC & Case HD\\
  \hline
 \multirow{4}{*}{\rotatebox{90}{Uniform}} & $|G1|=50$, $|G2|=200$ & $0.6560 (0.0294)$ & $0.9077 (0.0180)$& $0.9066 (0.0168)$	& $0.9234 (0.0165)$ \\
  & $|G1|=10$, $|G2|=40$ & $0.6924 (0.0603)$ & $0.9076 (0.0381)$ & $ 0.9124 (0.0383)$ & $0.9378 (0.0286)$\\	
   & $|G1|=|G2|=125$ &$0.6585 (0.0282)$ & $0.9303 (0.0154)$ & $0.9313 (0.0158)$ & $0.9448 (0.0154)$	\\
    & $|G1|=|G2|=25$ &$0.6964 (0.0534)$ & $0.9356 (0.0337)$ & $0.9348 (0.0312)$ & $0.9560 (0.0233)$\\
   \hline
       \hline 
    \multirow{4}{*}{\rotatebox{90}{Normal}} & $|G1|=50$, $|G2|=200$ & $0.6576 (0.0334)$ & $0.8984 (0.0216)$& $0.9080 (0.0174)$	& $0.9337 (0.0244)$ \\
  & $|G1|=10$, $|G2|=40$ & $0.7008 (0.0692)$ & $0.9000 (0.0421)$ & $ 0.9184 (0.0334)$ & $0.9506 (0.0342)$\\	
   & $|G1|=|G2|=125$ &$0.6636 (0.0301)$ & $0.9232 (0.0180)$ & $0.9316 (0.0139)$ & $0.9560 (0.0196)$	\\
    & $|G1|=|G2|=25$ &$0.6914 (0.0666)$ & $0.9248 (0.0390)$ & $0.9404 (0.0261)$ & $0.9648 (0.0294)$\\
   \hline 
       \hline 
     \multirow{4}{*}{\rotatebox{90}{Log-Normal}} & $|G1|=50$, $|G2|=200$ & $0.6691 (0.0338)$&	$0.8533 (0.0202)$	&$0.9007 (0.0197)$&	$0.9534 (0.0117)$\\
  & $|G1|=10$, $|G2|=40$ & $0.6878 (0.0737)$	& $0.8540 (0.0474)$	& $0.9062 (0.0466)$	& $ 0.9496 (0.0299)$\\	
   & $|G1|=|G2|=125$ & $0.6740 (0.0298)$	&  $0.8911 (0.0178)$	& $ 0.9228 (0.0178)$	& $0.9678 (0.0099)$	\\
    & $|G1|=|G2|=25$ & $0.6844 (0.0605)$&	$0.8932 (0.0449)$	& $0.9292 (0.0349)$	& $0.9702 (0.0237)$\\
   \hline 
       \hline 
     \multirow{4}{*}{\rotatebox{90}{Mix}} & $|G1|=50$, $|G2|=200$ & $0.6314 (0.0445)$ &	$0.8819 (0.0208)$ &	$0.8967 (0.0168)$	& $0.9274 (0.0165)$\\
& $|G1|=10$, $|G2|=40$ & $0.6456 (0.0659)$	& $0.8866 (0.0410)$ &	$0.9086 (0.0352)$	& $0.9272 (0.0443)$\\
& $|G1|=|G2|=125$ & $0.6393 (0.0408)$ &	$0.9126 (0.0173)$	& $0.9208 (0.0167)$& 	$0.9569 (0.0188)$\\
 & $|G1|=|G2|=25$ &  $0.6328 (0.0726)$	&$0.9182 (0.0366)$&	$0.9248 (0.0399)$	&$0.9574 (0.0303)$ \\
   \hline   
\end{tabular}}
\caption{Mean and standard deviation of the proportion of well classified units in the learning set for the histogram-valued variables.}\label{table1}
\end{center}
\end{table}

\begin{table}[h!]
\begin{center}
{\scriptsize
\begin{tabular}{c|p{.18\textwidth}|l|cccc}
  \hline
  & Test set & Learning set &  Case HA  &   Case HB &   Case HC & Case HD\\
  \hline
 \multirow{4}{*}{\rotatebox{90}{Uniform}} & $|G1|=200$  & $|G1|=50$, $|G2|=200$ & $0.6277(0.0109)$&	$0.9055(0.0027)$	&$0.8941(0.0027)$	&$0.9122(0.0139)$ \\
  &  $|G2|=800$& $|G1|=10$, $|G2|=40$ & $0.6095(0.0195)$	&$0.8995(0.0099)$	&$0.8908(0.0050)$&	$0.9058(0.0177)$\\	 
  \cline{2-7}
   &  \multirow{2}{*}{$|G1|=|G2|=500$} & $|G1|=|G2|=125$ & $0.6226(0.0085)$	& $0.9335(0.0028)$	&$0.9242(0.0033)$	&$0.9361(0.0099)$	\\
    & & $|G1|=|G2|=25$ &$0.6108(0.0188)$	&$0.9267(0.0085)$&	$0.9210(0.0071)$&	$0.9296(0.0128)$\\
   \hline
     \hline
    \multirow{4}{*}{\rotatebox{90}{Normal}} &  $|G1|=200$  & $|G1|=50$, $|G2|=200$ & $0.6421(0.0189)$&	$0.9103(0.0044)$	&$0.8964(0.0049)$	&$0.9263(0.0216)$ \\
  & $|G2|=800$&  $|G1|=10$, $|G2|=40$ & $0.6254(0.0245)$	&$0.9024(0.0107)$&	$0.8926(0.0087)$	&$0.9232(0.0231)$\\	
    \cline{2-7}
   &   \multirow{2}{*}{$|G1|=|G2|=500$}  &  $|G1|=|G2|=125$ & $0.6417(0.0188)$&	$0.9343(0.0031)$&	$0.9294(0.0030)$	&$0.9482(0.0193)$	\\
    &  &  $|G1|=|G2|=25$ &$0.6235(0.0209)$&	$0.9248(0.0116)$	&$0.9253(0.0068)$	&$0.9418(0.0218)$\\
   \hline
     \hline 
     \multirow{4}{*}{\rotatebox{90}{Log-Normal}} &  $|G1|=200$   & $|G1|=50$, $|G2|=200$ & $0.5611(0.0162)$ &	$0.8016(0.0056)$ &	$0.8898(0.0048)$ &	$0.8718(0.0076)$\\
  & $|G2|=800$& $|G1|=10$, $|G2|=40$ & $0.5432(0.0283)$	 &$0.7949(0.0109)$	 &$0.8889(0.0106)$	 &$0.8600(0.0288)$\\
    \cline{2-7}	
   &  \multirow{2}{*}{$|G1|=|G2|=500$}  & $|G1|=|G2|=125$ & $0.6118(0.0175)$ &	$0.8488(0.0055)$ &	$0.9127(0.0036)$ &	$0.8948(0.0056)$	\\
    & &  $|G1|=|G2|=25$ & $0.5943(0.0248)$	 & $0.8432(0.0102)$  &	$0.9093(0.0079)$ &	$0.8858(0.0205)$\\
   \hline 
     \hline
     \multirow{4}{*}{\rotatebox{90}{Mix}} & $|G1|=200$ & $|G1|=50$, $|G2|=200$ & $0.5990(0.0332)$	&$0.8552(0.0078)$&	$0.8655(0.0059)$	&$0.8875(0.0058)$\\
& $|G2|=800$ & $|G1|=10$, $|G2|=40$ & $0.5827(0.0429)$	&$0.8543(0.0150)$&	$0.8639(0.0106)$&	$0.8737(0.0384)$\\
  \cline{2-7}
&   \multirow{2}{*}{$|G1|=|G2|=500$}  & $|G1|=|G2|=125$ & $ 0.6247(0.0262)$	 &$0.8883(0.0069)$&	$0.8898 (0.0038)$	&$0.9157 (0.0181)$	\\
 & &  $|G1|=|G2|=25$ &  $0.6022(0.0299)$&	$0.8834(0.0145)$&	$0.8873 (0.0084)$&	$0.9022 (0.0254)$	 \\
   \hline   
\end{tabular}}
\caption{Mean and standard deviation of the proportion of well classified units in the test sets for the histogram-valued variables.}\label{table2}
\end{center}
\vspace{1cm}
\end{table}

\pagebreak[4]

\begin{table}[h!]
\begin{center}
{\scriptsize
\begin{tabular}{l|cccccc}
  \hline
 & Case IA  & Case IB &  Case IC  & Case ID  &Case IE  &Case IF \\
  \hline
$|G1|=50$, $|G2|=200$ & $0.7432 (0.0298)$ 	 &	 $0.9342(0.0163)$	 &	$0.7393 (0.0298)$	& $0.9314(0.0139)$ &  $0.7526 (0.0248)$ & $0.9408 (0.0135)$	\\
$|G1|=10$, $|G2|=40$ & $0.7588 (0.0622)$	 &	$0.9384 (0.0320)$	 &	$0.7638 (0.0597)$	& $0.9330(0.0354)$	& $0.7846 (0.0497)$ &$0.9536 (0.0244)$\\
$|G1|=|G2|=125$ & $0.7533(0.0296)$	 &	$0.9545(0.0142)$	 &	$0.7496(0.0284)$	 &	$0.9534(0.0116)$	&$0.7669 (0.0255)$  &$0.9598(0.0117)$ \\
$|G1|=|G2|=25$ & $0.7678 (0.0653)$	 &	$0.9614(0.0230)$	&   $0.7744(0.0591)$	    &	    $0.9560(0.0256)$ & $0.7850 (0.0057)$ & $0.9716(0.0230)$\\
   \hline
\end{tabular}}
\caption{Mean and standard deviation of the proportion of well classified units in the learning set for the interval-valued variables.}\label{table3}
\end{center}
\end{table}

\begin{table}[h!]
\begin{center}
{\scriptsize
\begin{tabular}{p{.18\textwidth}|l|cccccc}
  \hline
   Test set & Learning set &  Case IA  &   Case IB &   Case IC & Case ID &   Case IE & Case IF\\
  \hline
  $|G1|=200$  & $|G1|=50$, $|G2|=200$ & $0.7379 (0.0061)$  &	$0.9330(0.0020)$	&$0.7379 (0.0088)$	&$0.9264(0.0032)$  & $0.7387 (0.0087)$ & $0.9343 (0.0054)$\\
$|G2|=800$& $|G1|=10$, $|G2|=40$ &$0.7193 (0.0338)$	&$0.9273 (0.0060)$	&$0.7173 (0.0334)$&	$0.9199(0.0100)$ & $0.7256 (0.0254)$ & $0.9316 (0.0054)$ \\	 
  \hline
  \multirow{2}{*}{$|G1|=|G2|=500$} & $|G1|=|G2|=125$ & $0.7557(0.0050)$	& $0.9528(0.0022)$	&$0.7679(0.0074)$	&$0.9468(0.0020)$ & $0.7533 (0.0095)$	 & 	$0.9564(0.0050)$\\
& $|G1|=|G2|=25$ &$0.7399 (0.0227)$	&$0.9483(0.0052)$&	$0.7461(0.0291)$&	$0.9438(0.0042)$ & $0.7413 (0.0253)$ & $0.9529(0.0082)$\\
\hline
\end{tabular}}
\caption{Mean and standard deviation of the proportion of well classified units in the test sets for the interval-valued variables.}\label{table4}
\end{center}
\end{table}

\end{landscape}

\end{document}